\documentclass[conference,compsoc]{IEEEtran}

\usepackage[utf8]{inputenc}
\usepackage[T1]{fontenc}
\usepackage{xcolor}
\usepackage{pifont}
\usepackage{amssymb}

\usepackage{amsmath,amsfonts}
\usepackage{float}
\usepackage{amsthm}
\usepackage{stfloats}
\allowdisplaybreaks
\usepackage{graphics,graphicx, color}
\usepackage{bm} 
\usepackage[linesnumbered,ruled,vlined]{algorithm2e}
\usepackage{threeparttable}

\usepackage{algorithmic} 
\usepackage{tikz}
\usepackage{booktabs}
\usetikzlibrary{calc}
\usepackage[hidelinks]{hyperref}
\usepackage[capitalize,nameinlink]{cleveref}
\usepackage{multirow}
\usepackage{subcaption}

\usepackage{tikz}
\usetikzlibrary{matrix,positioning,decorations.pathreplacing}
\usepackage{booktabs}
\usetikzlibrary{automata,positioning,arrows}

\usepackage{tikz}
\usetikzlibrary{shapes.geometric, arrows.meta, positioning, calc, backgrounds, fit}

\tikzset{
    server/.style={cylinder, draw, shape border rotate=90, minimum height=1.4cm, minimum width=1.2cm, aspect=0.2, fill=gray!10},
    client/.style={rectangle, draw, rounded corners, minimum height=1.2cm, minimum width=1.5cm, fill=blue!5},
    arr/.style={-{Stealth[scale=1.0]}, thick},
    attack_arr/.style={-{Stealth[scale=1.0]}, thick, red}, %
    note/.style={font=\scriptsize, align=center}
}

\newtheorem{thm}{Theorem}
\numberwithin{thm}{section}

\newtheorem{rmk}[thm]{Remark}
\newtheorem{cor}[thm]{Corollary}

\Crefname{defi}{Definition}{Definition}
\crefname{lem}{Lemma}{Lemma}
\crefname{prop}{Proposition}{Proposition}
\crefname{rmk}{Remark}{Remarks}
\Crefname{rmk}{Remark}{Remarks}
\crefname{cor}{Corollary}{Corollaries}
\Crefname{cor}{Corollary}{Corollaries}

\newcommand{\va}{{\bm a}}
\newcommand{\vb}{{\bm b}}

\newcommand{\ve}{{\bm e}}

\newcommand{\vm}{{\bm m}}

\newcommand{\vp}{{\bm p}}

\newcommand{\vr}{{\bm r}}
\newcommand{\vs}{{\bm s}}
\newcommand{\vt}{{\bm t}}
\newcommand{\vu}{{\bm u}}

\newcommand{\Z}{{\mathbb Z}}

\newcommand{\calO}{\mathcal{O}}

\newcommand{\calR}{\mathcal{R}}
\newcommand{\calS}{\mathcal{S}}
\newcommand{\calT}{\mathcal{T}}

\newcommand{\tsf}[1]{\textsf{{#1}}}

\newcommand{\pk}{\ensuremath{\mathsf{pk}}}
\newcommand{\sk}{\ensuremath{\mathsf{sk}}}
\newcommand{\evk}{\ensuremath{\mathsf{evk}}}
\newcommand{\ct}{\ensuremath{\mathsf{ct}}}

\SetKwInOut{Parameter}{Parameter}

\DeclareMathOperator{\poly}{poly}

\newcommand{\round}[1]{\ensuremath{\left\lfloor #1 \right\rceil}}

\newcommand{\set}[1]{\ensuremath{\{ {#1} \}}}

\begin{document}
\title{Reaction Attack on CKKS Deployments\\ with Asymptotically Optimal Query Complexity{\IEEEauthorrefmark{1}}}

\author{ 
\IEEEauthorblockN{Song Bian} 
\IEEEauthorblockA{Beihang University \\sbian@buaa.edu.cn}
\and 
\IEEEauthorblockN{Hyeongmin Choe} 
\IEEEauthorblockA{University of Luxembourg \\hyeongmin.choe@uni.lu}
\and 
\IEEEauthorblockN{Jiseung Kim}
\IEEEauthorblockA{Jeonbuk National University \\jiseungkim@jbnu.ac.kr}
\and 
\IEEEauthorblockN{Changmin Lee\IEEEauthorrefmark{2}} 
\IEEEauthorblockA{Korea University \\changminlee@korea.ac.kr}
}

\maketitle

\renewcommand{\thefootnote}{\fnsymbol{footnote}}
\footnotetext[1]{{Authors are listed in alphabetical order; see \protect\url{https://www.ams.org/profession/leaders/culture/JointResearchandItsPublicationfinal.pdf}.}}
\footnotetext[2]{Corresponding author.}
\renewcommand{\thefootnote}{\arabic{footnote}}

\begin{abstract}
We study a protocol-level reaction leakage in CKKS deployments where public coefficient-to-slot (C2S) evaluation keys are combined with a client's 1-bit validity response. Under this deployment surface, a malicious evaluator can recover a ternary secret key using asymptotically optimal $O(\mathcal{H})$ reaction queries, where $\mathcal{H}$ is the entropy of the secret key.

The attack does not break RLWE or CKKS IND-CPA security. Instead, it exploits the interaction between client-side validity reactions and public homomorphic linear transformations. A shift of a public encryption key yields a ciphertext whose plaintext is a function of the secret key; the public ${\sf C2S}$ layer then repacks this information to construct a ciphertext whose validity depends strictly on whether a targeted secret coefficient matches a guessed value. Furthermore, each submitted query is independently randomized with a fresh encryption of zero, so that each resulting malicious ciphertext appears entirely distinct, thereby allowing the attack to simply evade trivial detection mechanisms on the client side. 

We implemented the full end-to-end attack in Lattigo v6.2.0 and recovered the full secret key across four parameter sets drawn from standard CKKS presets. 
For uniform ternary secrets, the attack uses
$O(N)$ queries and is within a factor $2/\log_2 3 < 2$ of the oracle-only entropy lower bound. For sparse ternary secrets, a tailored variant further reduces the number of queries in practice. Our results show that observable validity reactions form a practical key-recovery channel in CKKS deployments exposing this surface, even when the underlying cryptographic assumptions remain intact.

\end{abstract}

\section{Introduction}\label{Sec: Introduction}

Fully Homomorphic Encryption (FHE)~\cite{Gen09} enables computations to be performed directly on encrypted data. In a typical outsourcing scenario, a client encrypts its input, sends the ciphertext to an untrusted server, and later decrypts the returned ciphertext to obtain
the result of the computation, without revealing the underlying data to
the server.

This capability makes FHE a natural foundation for secure outsourced
computation, including privacy-preserving machine learning and other
cloud-based data-processing tasks~\cite{AB24, CGL+20, KHB+21, KJL+22, LKL+22, LJ21, SBL22, SZL+18}. 
Although FHE implementations have grown increasingly practical, the security of these real-world deployments is no longer determined only by the hardness of the underlying cryptographic assumptions. 
It additionally hinges on how the resulting ciphertexts are returned, checked, and handled by the client.

\medskip\noindent\textbf{Attacks on Secure Outsourced Computation.}
Attacks on secure outsourced computation based on FHE can be broadly divided into two classes. The first class targets the cryptographic primitives themselves, independent of the specific application. This includes the cryptanalysis of the underlying Ring-Learning-With-Errors (RLWE) assumptions, a foundational hardness assumption widely used in lattice-based cryptography, or exploiting the gap in theoretical security notions, most notably, the fact that standard FHE schemes guarantee Indistinguishability under Chosen-Plaintext Attack (IND-CPA) but are inherently vulnerable to Indistinguishability under Chosen-Ciphertext Attack (IND-CCA) attacks.

In contrast, this paper focuses on the second class of attacks, which emerge from the deployment and protocol layers surrounding homomorphic encryption. 
While much of the applied FHE literature assumes an honest-but-curious server, this assumption is often overly optimistic for practical deployment, where a server might deviate from the protocol. 
In such scenarios, a malicious evaluator can actively manipulate the computation, submitting adversarially generated ciphertexts to observe how the client reacts--whether they accept, reject, retry, or leak information through decrypted outputs.

\medskip\noindent\textbf{Reaction Attacks.}
This specific vulnerability is exploited by reaction attacks~\cite{Ble98, CCCM22, CCCM24, CGG16, DDA13, HGS99, HSJ09, LMSV12}, which turn a client's visible behavior--such as acceptance (accept or reject), rejection (retry or not), or error handling--into a one-bit decryption oracle. 
Although each individual response leaks only a single bit of information, a sequence of adaptively chosen queries can amplify this leakage into substantial information about the secret key. 

To illustrate this setting, depicted in Figure~\ref{fig:reaction_attack_overview}, consider a client ${\sf C}$ that outsources a computation $\Phi$ to a server ${\sf S}$, typically a cloud service provider. 
The client encrypts a message $m$ under the Cheon--Kim--Kim--Song (CKKS) scheme~\cite{CKKS17}, an FHE scheme specialized for real and complex number vectors, sends the ciphertext $\mathsf{Enc}(m)$ to ${\sf S}$, and expects to receive an evaluated ciphertext decrypting to an approximation of $\Phi(m)$. In an idealized model, the client simply decrypts the returned ciphertext to obtain the result. 
In practice, however, a client \emph{seldom} accepts the decrypted output \emph{unconditionally}. 
In practice, deployments (whether or not they use FHE) often perform application-layer admissibility checks (e.g., range or magnitude tests) before consuming a result. This kind of input validation is a well-established defensive-programming technique~\cite{owasp-input-validation}, and reaction attacks beginning with Bleichenbacher~\cite{Ble98} and continuing thereafter have repeatedly demonstrated that the resulting accept/reject behavior can leak secret information.
For instance, a machine-learning pipeline might verify whether the decoded class logits, the pre-softmax scores, fall within a valid range, or a protocol might check for arithmetic overflow. 
Whenever the client's application-level reaction triggered by a failed check (e.g., an implicit retry or an explicit abort) becomes observable to the server, this check effectively transforms into a reaction oracle.

The crucial point is that such an oracle does not reveal plaintexts, nor does it require the adversary to break the underlying RLWE assumption or the IND-CPA security of CKKS. 
It reveals strictly whether a submitted ciphertext satisfies the client's validity predicate. 
In this scenario, a malicious server ${\sf S}$ can submit a manipulated ciphertext $(\va, \vb)$ and observes only whether the client-side predicate accepts or rejects the decoded output. 
For example, if ${\sf C}$ proceeds without requesting a re-computation, the adversary learns that the predicate was accepted; if ${\sf C}$ retries or raises an error, the adversary learns it failed. 
Thus, the validity reaction acts as a one-bit decision oracle on the decoded slot magnitudes. 
Yet, as we demonstrate, this single bit is sufficient to fully recover the CKKS secret key under a natural deployment surface.

Even such a restricted oracle can leak the secret key. 
To see this, let $\vs$ be the client's secret key, where a valid ciphertext $(\va, \vb)$ satisfies the standard decryption relation $\vb - \va\vs = \vm + \mathbf{e} \bmod q$ for a plaintext $\vm$, a small encryption noise $\mathbf{e}$, and a ciphertext modulus $q$.
Suppose a malicious server ${\sf S}$ attempts to guess the secret key by sampling a candidate ring element $\vu$ and a sufficiently large public perturbation $\vt$. 
The server then sends the modified ciphertext $(\va' = \va + \vt, \vb' = \vb + \vu\vt)$ to ${\sf C}$.

If the guess is correct (i.e., $\vu = \vs$), the decryption relation remains unchanged since $\vb' - \va'\vs = \ve \bmod q$, meaning the modified ciphertext stays within the admissible range. 
If the guess is incorrect (i.e., $\vu \neq \vs$), the decryption residual contains the additional term $(\vt(\vu - \vs) + \ve \bmod q)$, which exceeds the admissibility threshold with high probability due to the magnitude of $\vt$. 
A reaction oracle can therefore distinguish whether a guessed ring element matches the secret key.

While this simple example demonstrates why validity reactions leak information about $\vs$, it is exponentially inefficient, equivalent to a brute-force search over a secret space (e.g., of size $3^N$ for typical uniform ternary secrets, where $N$ is the ring dimension). 

Crucially, reaction attacks are not captured by standard security notions for FHE, such as IND-CPA, IND-CCA, or the intermediate $\text{IND-CPA}^{\text{D}}$~\cite{LM21, LMSS22}, which allows restricted decryption queries.
The adversary never obtains the plaintexts themselves, nor does it violate the hardness of the underlying RLWE problem. 
It merely observes the binary validity outcomes induced by the client's application-level correctness checks.

\begin{figure}[t!]
\centering
\begin{tikzpicture}[node distance=4.5cm, font=\sffamily\small]

    \node[server] (server) {Server $\mathsf{S}$};
    \node[client, right=3.5cm of server, yshift=0.1cm] (client) {Client $\mathsf{C}$};

    \node[above=0.05cm of server, font=\bfseries] {\small \sf pk};
    \node[above=0.15cm of client, font=\bfseries] {\small \sf sk};

    \draw[arr] ($(server.east)+(3.4,0.65)$) -- node[above, note] {1. $\mathsf{ct}_\mathrm{in} = \mathsf{Enc}(\mathbf{m})$} ($(server.east)+(0.1,0.65)$);

    \draw[arr, gray, dashed] ($(server.east)+(0.1,0.05)$) -- node[above, note] {2. $\mathsf{ct}_\mathrm{out} = \mathsf{Eval}_{\sf pk}(\Phi, \mathsf{ct}_\mathrm{in})$} ($(server.east)+(3.4,0.05)$);

    \draw[attack_arr] ($(server.east)+(0.1,-0.5)$) -- node[above, note] {2. Replaced by $\mathsf{ct}_\mathrm{adv}$} ($(server.east)+(3.4,-0.5)$);

    \node[note, draw, rounded corners, fill=white, below=0.4cm of client] (val) {3. Deploy.-lvl. check\\ (modeled as $\calO_\mathrm{react}$)};
    \draw[arr, gray] (client.south) -- (val.north);

    \node[note, draw, red, thick, fill=red!5] at (val -| server) (recovery) {5. {\sf sk} recovery};
    
    \draw[attack_arr] ($(val.west)+(-0.1,0)$) -- node[above, note] {4. Reaction leakage} node[below, note] {1 bit (valid or not)} ($(recovery.east)+(0.1,0)$);

\end{tikzpicture}
\caption{Conceptual Sketch of the Reaction Attack Scenario.}
\label{fig:reaction_attack_overview}
\end{figure}
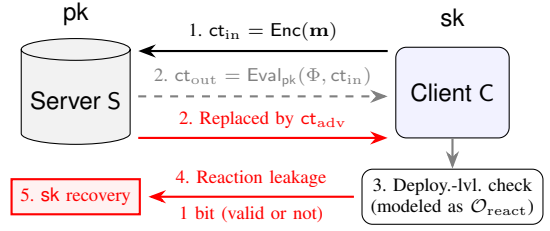

\subsection{Limitations of Prior Reaction Attacks}\label{subsec:prior attack}

To the best of our knowledge, only a few works have explicitly
considered reaction attacks on FHE schemes within the secure outsourced
computation model~\cite{CCCM22,CCCM24,KK25,LMSV12,ZPS12}. These works have
played an important role in showing that client-visible validity
responses can become a key-recovery channel. However, they have
primarily focused on exact FHE schemes, such as Gentry-type integer
schemes~\cite{Gen09,ZPS12}, TFHE~\cite{CGGI16},
FHEW~\cite{DM15}, BGV~\cite{BGV14}, and BFV~\cite{Bra12,FV12};
extending them to the approximate CKKS setting is the gap we address.

While crucial in establishing the threat of reaction-based key recovery, these prior attacks were designed for exact schemes and do not transfer directly to RLWE-based approximate FHE such as CKKS~\cite{CKKS17}, where the noisy decoding blurs the deterministic error separation they exploit. The practical difference most relevant to our setting is how queries are formed, discussed below. 

\medskip\noindent\textbf{Detectability via Ciphertext Reuse.}
Earlier reaction attacks~\cite{CCCM22,CCCM24,LMSV12} heavily rely on repeatedly submitting minor variations of a fixed target ciphertext to the client. 
For instance, in attacks on TFHE/BGV, the adversary fixes a component of the ciphertext $(\va,\vb)$ and issues numerous queries on nearby variants $(\va,\vb+\delta)$. 
By incrementally adjusting the perturbation $\delta$ via a binary-search-style procedure, the adversary aims to pinpoint the exact threshold that triggers a change between a valid and invalid response, thereby exposing the underlying noise. 
This methodology creates linkable correlations across queries~\cite{KK25}: a client that inspects its access logs may notice many requests in which one ciphertext component is held fixed while the other is systematically varied. Our attack instead re-randomizes every query (\Cref{sec: threat-model}), so it does not rely on such a reuse pattern.

\medskip
These fundamental limitations raise a concrete and challenging question:
\begin{center}
\textit{
Can single-bit validity responses enable efficient\\
secret-key recovery in FHEs using only\\
format-compliant ciphertexts without ciphertext reuse?
}
\end{center}

\subsection{Our Result}\label{sec:this-work}

\begin{table*}[h!]
\centering
\small
\begin{tabular}{c|cccc||cccc}
\toprule
& \multicolumn{4}{c||}{\textbf{Uniform Ternary Secret (\Cref{alg:on-attack})}} & \multicolumn{4}{c}{\textbf{Sparse Ternary Secret (\Cref{alg:algo_sparse})}} \\
\midrule
Ref. & $(\log N, \log q_0, \log \Delta)$ & Avg. $\#Q$ & $\#Q/H(\vs)$ & E2E Time & $(\log N, h)$ & Avg. $\#Q$ & $\#Q/H(\vs)$ & LLL Time \\
\midrule
\multirow{4}{*}{Lattigo \cite{lattigo}\protect\footnotemark} 
 & $(12, 35, 30)$ & $6{,}832$   & $1.052$ & $30$\,s & $(12, 32)$    & $338$     & $1.134$ & $0.4$\,s \\
 & $(12, 37, 32)$ & $6{,}826$   & $1.051$ & $30$\,s & $(12, 64)$    & $598$     & $1.117$ & $3.3$\,s \\
 & $(14, 45, 40)$ & $27{,}313$  & $1.052$ & $506$\,s & $(12, 128)$  & $1{,}079$ & $1.142$ & $~~1{,}521$\,s~\,\,\, \\
 & $(16, 55, 50)$ & $109{,}232$ & $1.052$ & $8{,}529$\,s & $(14, 128)$ & $1{,}469$ & $1.221$ & $~~1{,}521$\,s$^{**}$ \\
\cline{1-5}
Guidelines~\cite{BCC+24} 
 & $(16, 60, 58)$ & $109{,}219$ & $1.051$ & $68{,}408$\,s$^*$ & $(16, 128)$   & $1{,}801$ & $1.234$ & $~~1{,}521$\,s$^{**}$ \\
\bottomrule
\end{tabular}
\caption{
Overview of Empirical End-to-End Attack (achieving a 100\% success rate across all evaluated cases). 
The baseline attack (Algorithm 1) reports average query-count that tracks closely the theoretical $N+h \approx 5N/3$ for uniform ternary secrets. End-to-end attack wall time is reported in seconds.
The variant targeting sparse ternary secrets (Algorithm 2) of fixed Hamming weights ($h \ll N$) compresses the required oracle queries; then the LLL-based sign-recovery completes the secret key recovery, whose cost is reported in seconds. 
$\#Q$ denotes the number of queries, whose theoretical lower bound is Shannon entropy $H(\vs)$.
More details can be found in \Cref{sec:experiment}.
}
\label{tab: number of queries}
\end{table*}
\footnotetext{
Parameters for Lattigo are described at \url{https://pkg.go.dev/github.com/ldsec/lattigo/v2/ckks}.
}

We show that a single-bit validity reaction suffices to recover a CKKS secret key~\cite{CHKKS18, CHKKS19, CKKS17} without reusing a fixed target ciphertext.

\medskip\noindent\textbf{Contributions.}
Our contributions are four-fold. 
\emph{(i) Optimality:} full ternary key recovery in $N+|\mathcal{S}|\le 2N$ \emph{deterministic} queries, within a factor $1.05$--$1.26$ of the entropy lower bound (\Cref{thm: lower-bound}), and $O(h\log(N/h))$ for sparse secrets.
\emph{(ii) A single-bit oracle:} each query returns only one validity bit per ciphertext, not per slot.
\emph{(iii) Reuse-resistant queries:} every query is freshly re-randomized and can even be added to an honest circuit output.
\emph{(iv) Approximate arithmetic:} the first realization for CKKS, via a public C2S localization and a slot-margin noise analysis (\Cref{thm: main}).

We state our result informally in \Cref{informal result}, with formal versions in \Cref{thm: main} (uniform ternary), \Cref{thm: sparse} (sparse ternary), and the matching lower bound in \Cref{thm: lower-bound}.

\medskip\noindent\textbf{Public Material Assumed by the Attack.}
The attack requires only standard public evaluation material: the public encryption key $\pk$
and the public rotation keys needed to homomorphically evaluate a Coefficient-to-Slot (C2S) transform. These rotation keys are standard CKKS evaluation keys implementing ring automorphisms, frequently referred to as Galois keys. While widely deployed for CKKS bootstrapping~\cite{BMTH21,CHKKS18,HS18}, the attack does not strictly require a typical bootstrapping setup. For instance, the attack does not depend on a fixed, pre-shipped key bundle: a reduced key set or on-the-fly key generation~\cite{CKP26,LLKN23} similarly suffices. Furthermore, C2S can be evaluated even in non-bootstrappable settings using keys at lower moduli, as in homomorphic matrix computations~\cite{BCH26,Park25} where only a few Galois keys are provided.\footnote{as few as $2$--$3$ keys are provided in~\cite{BCH26,Park25}, including a single-slot rotation and conjugation, sufficient to generate the full automorphism group needed for C2S.}

\begin{thm}[Informal Result]
\label{informal result}
Let $N$ be the ring dimension. 
For the CKKS scheme with a ternary secret key $\vs$, public encryption key $\pk$, and public rotation keys for the Coefficient-to-Slot (C2S) transform, the secret key $\vs$ can be fully recovered in time polynomial in $N$ using at most $2N$ reaction-oracle queries. 

Furthermore, if the secret key $\vs$ is restricted to a sparse ternary distribution with $h$ non-zero coefficients, 
then the secret key $\vs$ can be fully recovered in time polynomial in $N$ and $h$, using $O(h\log(N/h))$ reaction-oracle queries.
\end{thm}

We note that in the sparse secret setting, which yields a Shannon entropy of $H(\vs) = O(h\log (N/h))$ for $h \ll N$.\footnote{Since $|\text{key space}| = \binom{N}{h} \cdot 2^h \le \left({2eN}/{h}\right)^h$ holds for $1 \le h \le N$.} 
As a result, our query complexity asymptotically attains the information-theoretic lower bound, yielding optimal query complexity. 

\medskip\noindent\textbf{Overview of the Attack.}
We implemented the full end-to-end attack in Lattigo~v6~\cite{lattigo}
across various standard parameter regimes ranging from $\log_2 N = 12$ to $16$, successfully recovering the secret key with a 100\% success rate in all evaluated cases (see \Cref{sec:experiment}). 
These ring dimensions and level-zero moduli (also known as bottom moduli $q_0$) strictly match standard CKKS parameters. 
\Cref{tab: number of queries} summarizes the theoretical query-count bounds and the empirical query-counts, with the  across the parameter sets, demonstrating full consistency with recent security recommendations~\cite{BCC+24} and commonly used library presets~\cite{lattigo}.

Crucially, our evaluation demonstrates that the required oracle calls scale strictly with the structural entropy of the secret key distribution. 
For standard uniform ternary secrets, full key recovery is consistently achieved within $1.7N$ queries: exactly $N$ queries are utilized to recover the support $\mathcal{S} = \{j: s_j \neq 0\}$, and $h = |\mathcal{S}|$ queries to determine the exact signs. 

More devastatingly, under standard sparse-ternary-secret regimes (e.g., fixed-weights of $h \in \{32, 64, 128, 192, 256\}$), our attack employs a bisection-based intersection search to perfectly recover the support in only $O(h\log_2 (N/h))$ queries. 
Once the support is identified, the attack exploits the underlying LWE relation to map the remaining sign-recovery task into a low-dimensional LWE instance. 
This allows us to bypass further oracle queries entirely and solve the instance via standard lattice reduction (LLL), in a matter of seconds to minutes (e.g., $0.4$\,seconds for $h=32$, and $3.3$\,seconds for $h=64$). According to the lattice estimator~\cite{APS15}, this dimensional collapse reduces the theoretical hardness to only $33$--$43$ bits for $h=32$--$1024$, which is well below the standard 128-bit security.

\medskip\noindent\textbf{Artifact.}
Our artifacts are available at \url{https://github.com/JiseungKim90/reaction-attack}.
The released artifact comprises the attack codes, 
and the parameter-sweep scripts, built under Go~1.24 against Lattigo~\cite{lattigo}. 

\medskip\noindent\textbf{Scope of the Attack.}
Our attack is protocol-level rather than primitive-level, exploiting correctness validation leakage without breaking the underlying RLWE or CKKS IND-CPA security. It applies to deployments that expose validity-dependent responses to the evaluator. Because queries are independently randomized with fresh encryptions, the attack inherently evades simple ciphertext-reuse filters. We detail the exact scope in \Cref{sec: threat-model}, and discuss countermeasures in \Cref{sec: discussion}.

\medskip\noindent\textbf{Distinguishing from Recent IND-CPA$^\text{D}$ Attacks.}
A closely related line of work studies oracle-based key recovery in the IND-CPA$^\text{D}$ setting~\cite{BCD+24, CCPS24, GNS+24, LM21}. We stress that our attack is \emph{not} an IND-CPA$^\text{D}$ attack; rather, the two settings are \emph{incomparable}, differing along two axes.

First, the \emph{adversarial class} differs. IND-CPA$^\text{D}$ security constrains the adversary to \emph{legitimate} ciphertexts: outputs of the scheme's own encryption or honest evaluation, tracked together with their underlying plaintexts. Our attacker is a malicious evaluator that submits an \emph{arbitrary}, adversarially crafted ciphertext lying outside this legitimate evaluation tree; a deployment that enforces such legitimacy is therefore \emph{not} in our scope.

Second, the \emph{oracle interface} differs. The IND-CPA$^\text{D}$ attack on bootstrapping in~\cite{CCPS24} also uses only the correctness information, similar to our setting, but for \emph{every message slot}. On the other hand, our attack observes only a single validity bit per query, making it harder to pinpoint which slot caused the failure.

Consequently, the two lines of work address different deployment surfaces and leakage channels, and should be viewed as distinct and potentially cumulative threat models.

\subsection{Related Work}
\medskip\noindent\textbf{IND-CVA Security.}
Not specifically in the FHE context, however, Bleichenbacher~\cite{Ble98} and Hall, Goldberg, and Schneier~\cite{HGS99} introduced classical reaction attacks on public-key systems. Hu, Sun, and Jiang~\cite{HSJ09} later formalized such attacks through indistinguishability under a chosen-ciphertext verification attack (IND-CVA), where the adversary receives a one-bit validity response for adaptively submitted ciphertexts.
IND-CVA comes in a non-adaptive and an adaptive form; our attack is adaptive, since its sign-recovery phase (and the sparse-secret bisection) chooses each query as a function of previous validity bits.
This abstraction captures the kind of verification signal used in our attack.
Loftus et al.~\cite{LMSV12} showed that FHE can remain vulnerable to IND-CVA attacks even when it satisfies IND-CCA1 security; that is, achieving the stronger CCA1 notion does not by itself close the verification channel. Das, Dutta, and Adhikari~\cite{DDA13} further positioned IND-CVA between IND-CPA and IND-CCA security and related it to models such as IND-CCA1.5.

\medskip\noindent\textbf{Comparison with Prior Reaction Attacks on FHE.}
We summarize how our attack relates to prior reaction and reaction-style key-recovery attacks on FHE.
Zhang, Plantard, and Susilo~\cite{ZPS12} gave, to our knowledge, the earliest reaction attack on FHE-based outsourced computation, recovering the secret key of Gentry's scheme and of integer-based FHE. Their attack is enabled by the bootstrapping key of Gentry's framework, which publishes an encryption of each secret-key bit $\{\mathsf{Enc}(s_i)\}$; a validity reaction on the $i$-th such ciphertext then directly leaks $s_i$, and they observe that its influence on the failure rate is small enough to remain inconspicuous. Their reaction signal is probabilistic (a timing distribution $D_0$ vs.\ $D_1$) and their success requires repeated queries on a fixed message to build confidence.
Our setting is fundamentally different: CKKS (like BGV/BFV) publishes no encryption of the secret key, so the coordinate-wise leakage that~\cite{ZPS12} obtains from the bootstrapping key is simply unavailable, an obstruction also noted by~\cite[p.~462]{CCCM24}. We instead \emph{construct} a secret-bearing ciphertext from the public encryption key via a shift and a public C2S transform, and each query yields a deterministic validity bit rather than a statistical one.

Loftus et al.~\cite{LMSV12} exhibit an IND-CVA attack on a CCA1-secure SHE scheme, using coefficient-wise verification queries on exact ciphertexts.

Chaturvedi et al.~\cite{CCCM22, CCCM24} recover keys of the exact schemes TFHE/FHEW and BGV/BFV by repeatedly submitting perturbed variants of a target ciphertext and using the resulting decryption failures as the reaction signal; the BGV/BFV variants inject the perturbation after a modulus switch so that it affects the error rather than the message~\cite[p.~470]{CCCM24}. Their recovery is probabilistic and, by their own classification, non-adaptive (IND-CVA1): each secret coefficient is recovered by guessing its value and the error sign and submitting several samples per guess, so the query count carries a repetition factor (e.g., $\sim\!2rN$ for their $r=11$), whereas ours is deterministic, adaptive, and query-optimal ($N+|\mathcal{S}|$, or $O(h\log(N/h))$ for sparse secrets).

We note that the \emph{existence} of such a break is unsurprising: Walter~\cite{Walter22} observes that a ciphertext-verification oracle yields a CCA2-style attack that is inherently devastating to any IND-CPA-secure LWE encryption, and that for exact schemes it can be mounted without side-channel assistance. Our contribution is therefore not the existence of a reaction break but its optimal query complexity and its realization for the \emph{approximate} CKKS scheme, where the per-coordinate noise-extraction underlying such exact-scheme attacks does not directly apply. 

Kumazaki and Kaji~\cite{KK25} study the minimum number of oracle queries for TFHE reaction attacks and give a nearly optimal scheme; their optimality is specific to the LWE/TFHE reaction model, whereas our optimality is measured against the Shannon entropy of the ternary CKKS secret.

Finally, the ciphertext-crafting step behind our public-key shift is algebraically close to the plaintext-recovery technique of Chenal and Tang~\cite{ChenalTang14}; the distinction is that~\cite{ChenalTang14} assumes a decryption oracle returning exact plaintext bits on exact schemes, whereas we operate under a strictly weaker aggregate reaction bit, localize CKKS coefficients through C2S in approximate arithmetic, and re-randomize each query.

\medskip
\noindent\textbf{Other Security Analyses of FHE Schemes.}
Most security analyses of CKKS-like schemes consider different oracle surfaces. One line of work studies the IND-CPA$^\text{D}$ setting~\cite{CRBB24, CCPS24, GNS+24, LM21, LMSS22}, in which the adversary is restricted to \emph{legitimate} ciphertexts (honest encryptions and evaluations tracked with their plaintexts). Some of these attacks, notably~\cite{CCPS24}, already rely on a correctness bit rather than on recovered plaintexts; the distinction from our work is thus not the amount of leakage per query but the adversarial class (legitimate-only versus an arbitrarily crafted ciphertext) and the per-slot versus aggregate granularity (\Cref{sec:this-work}). Hwang et al.~\cite{HMSS25} study a differentially private homomorphic-evaluation notion for protecting the sender's message against a malicious receiver in CKKS-based protocols.

Another line of work studies implementation-level leakage, including side-channel and fault-injection attacks~\cite{AKP+22, CGG16, GB25, MNR+25}. These attacks rely on physical leakage, implementation behavior, or fault effects, whereas our adversary uses only public evaluation material and one observable decoded-slot validity bit.

This is the first polynomial-query reaction attack that recovers the full CKKS secret key from a decoded-slot validity oracle in deployments exposing a public C2S transform. The attack therefore fills a gap between IND-CPA$^\text{D}$ analyses, which grant the adversary decryption only on legitimate, evaluation-consistent ciphertexts, and implementation attacks, which require physical or implementation-specific leakage.

\section{Preliminaries}\label{Sec: Preliminaries}
Let $\Z$ be the set of all integers. 
For any integer $q \geq 2$, we define $\Z_q = [-q/2, q/2) \cap \Z$, where $z \bmod q$ denotes the reduction of an integer $z$ into this interval.
Unless specified otherwise, all logarithms are assumed to be base $2$.

For any set $S$, sampling $x$ from a uniform distribution over $S$ is denoted by $x \gets S$. When $x$ is drawn from a distribution $D$, it is written as $x \gets D$. For any positive integer $n$, the set $\set{0, \ldots, n-1}$ is denoted by $[n]$. 
For any $u \in \Z$, the rounding operation is denoted by $\round{u}$. If $\vr$ is a polynomial or a coefficient vector, $\round{\vr}$ is applied coefficient-wise. For a subset $S \subseteq \{0, \cdots, N-1\}$, let $\va|_S \in \mathbb{R}^{|S|}$ denote the restriction of a vector $\va \in \mathbb{R}^N$ to the coordinates in $S$. For a matrix $A \in \mathbb{R}^{M \times N}$, $\mathbf{A}|_S \in \mathbb{R}^{M \times |S|}$ denotes the submatrix consisting of the columns of $\mathbf{A}$ indexed by $S$.

Let $\calR = \Z[x]/\langle x^N + 1\rangle$, where $N$ is a power of two integer. For any integer $q \geq 2$, let $\calR_q = \calR / q\calR$. 
For a polynomial $\va=\sum_{j=0}^{N-1} a_j x^j \in \calR$, we write its coefficient embedding as $\mathrm{coeff}(\va):=(a_0,\ldots,a_{N-1})\in\mathbb{Z}^N$. 
When clear from the context, we abuse notation by using boldface letters to denote ring elements or, equivalently, their corresponding coefficient vectors.
We define the $\ell_\infty$-, $\ell_1$-, and $\ell_2$-norms of a polynomial $\va$ with its coefficients, by 
$\|\va\|_\infty := \max_{0\le j\le N-1} |a_j|$, 
$\|\va\|_1 := \sum_{j=0}^{N-1} |a_j|$, and 
$\|\va\|_2 := \sqrt{\sum_{j=0}^{N-1} a_j^2},$ respectively.
Throughout, an unsuperscripted norm $\|\cdot\|$ refers to this coefficient embedding; the superscripted norm $\|\vp\|_\infty^{\mathrm{slot}} := \|\mathsf{Decode}_\Delta(\vp)\|_\infty$ denotes the $\ell_\infty$-norm of the decoded CKKS slot vector (or, slot-norm) with a scale factor $\Delta$, where $\mathsf{Decode}_\Delta$ and the C2S transform are detailed in \Cref{subsec: c2s}.

Let $\calR$ and $q$ be as defined above, and let $\chi$ be an error distribution over $\calR$. For a secret $\vs \in \calR_q$, a sample from the Ring-Learning-With-Errors (RLWE) distribution with respect to $q$, $\vs$, and $\chi$ is generated by sampling $\va \gets \calR_q$ uniformly at random and $\ve \gets \chi$, and outputting the pair $(\va, \vb = \va \cdot \vs + \ve \bmod q) \in \calR_q^2$. 
The error distribution $\chi$ is typically chosen as a zero-centered discrete Gaussian distribution over $\calR$, where every coefficient is sampled independently from a discrete Gaussian distribution over $\Z$, denoted by $\chi_\Z$. Common choices for the secret distribution include a zero-centered discrete Gaussian, or narrower distributions such as the uniform ternary or fixed-weight ternary distributions.

\subsection{Fully Homomorphic Encryption (FHE)}
A fully homomorphic encryption (FHE) consists of four procedures.
\begin{itemize}
    \item $\mathsf{Setup}(1^\lambda)$: For the security parameter $\lambda$, the algorithm returns a tuple of keys $({\sf pk}, {\sf sk}, {\sf evk})$. 
    \item $\mathsf{Enc}(\pk, m)$: For the public key $\pk$ and a message $m$, the algorithm returns a ciphertext $\sf ct$.
    \item $\mathsf{Dec}(\sk, \sf ct)$: For the (master) secret key $\sk$ and a ciphertext $\sf ct$, it outputs a value $m'$.
    \item $\mathsf{Eval}({\sf pk}, {\sf evk}, \mathcal C, {\sf ct}_1, \ldots, {\sf ct}_k)$: For the evaluation key $\evk$, a circuit $\mathcal C$ and a bundle of ciphertext $\set{\sf ct}_{i=1}^k$, it returns another ciphertext $\sf ct'$.
\end{itemize}
We say FHE is correct if for every pair $(\pk, \sk)$ output by $\sf Setup$, for any bounded depth circuit $C$, and any message $m_1, \ldots, m_k$, the following holds with overwhelming probability:
\begin{align*}
    \mathsf{Dec}({\sf sk}, \mathsf{Eval}({\sf pk}, \mathcal C, {\sf ct}_1, \ldots, {\sf ct}_k)) = C(m_1, \ldots, m_k)\enspace,
\end{align*}
where ${\sf ct}_i \gets \mathsf{Enc}(\pk, m_i)$ for all $i \le k$.

\subsection{Ciphertext Format of CKKS Scheme}\label{subsec: ciphertext formats}
This section briefly reviews the ciphertext format for the CKKS FHE scheme \cite{CKKS17}.
Understanding this format is sufficient to follow our attack.
For any RLWE-based FHE scheme, every ciphertext is a pair of ring elements of the form $(\va, \vb) \in \calR_q^2$. Let $\vs$ be the secret polynomial in $\calR_q$ with small coefficients.

Given the primitive $2N$-th root of unity $\zeta$, define $\zeta_i = \zeta^{5^i}$ for an index $i$.
Then, the discrete Fourier Transformation map $\mathrm{DFT}_N: \mathbb{R}[x]/(x^N+1) \to \mathbb{C}^{N/2}$ defined by $m(x) \mapsto \bigl(m(\zeta_0), \ldots, m(\zeta_{N/2-1})\bigr)$ is an isomorphism, with its inverse $\mathrm{iDFT}_N$.
Based on these maps, we introduce the encoding and decoding algorithms used in the CKKS FHE scheme:
\begin{itemize}
    \item $\tsf{Encode}(\vec{m}, \Delta)$: Given $\vec{m} \in \mathbb{C}^{N/2}$ and a public scaling factor $\Delta$, it outputs $\vm = {\round{\Delta \cdot \mathrm{iDFT}_N(\vec{m})}} \in \calR$.
    \item $\tsf{Decode}(\vm, \Delta):$ Given $\vm \in \calR$ and a public scaling factor $\Delta$, it outputs $\vec{m} = \mathrm{DFT}_N(\Delta^{-1}\cdot\vm) \in \mathbb{C}^{N/2}$.
\end{itemize}

We refer to the resulting encoded message $\mathbf{m} \in \calR$ as the plaintext.
A CKKS ciphertext corresponding to the plaintext $\vm$ satisfies 
$\vb - \va \cdot \vs = \vm + \ve \bmod q$, where the noise polynomial $\ve$ is initially sampled from a discrete Gaussian distribution with a standard deviation $3.19$, 
and is subsequently updated by homomorphic operations. 

Let $\mathsf{ct} = (\va, \vb)$ be a ciphertext and $\sk = (-\vs, 1)$ be a secret key of CKKS.\footnote{For the sake of simplicity, we do not strictly distinguish between the secret polynomial $\vs$ and the tuple $\sk$, referring to both as the secret key.}
For the decryption correctness, CKKS libraries \cite{heaan2022, HElib, lattigo, PALISADE, sealcrypto} configure parameters such that
\begin{align}\label{eqn: size_ckks}
\|\langle \mathsf{ct}, \sk \rangle \bmod q\|_\infty
&= \|\vb-\va\cdot\vs \bmod q\|_\infty \nonumber\\
&= \|\vm+\ve\|_\infty \le \Delta \enspace.
\end{align}
for the (pre-determined) scaling parameter $\Delta$.
This requires the condition $\Delta < q$ to hold, ensuring the scaled message can be accurately recovered.
Libraries \cite{OpenFHE, lattigo} often configure parameters where $\log q \approx \log \Delta + 10$, providing a sufficient margin between
$\Delta$ and $q$.

\medskip{\noindent\textbf{Secret Distributions.}}
In recent FHE schemes \cite{BCC+24, BGV14, CKKS17, DM15, FV12} and libraries \cite{OpenFHE, heaan2022, HElib, lattigo, sealcrypto}, the secret key $\vs$ is usually drawn from a set of ternary polynomials $\calT^N[x] = \set{\sum_{i=0}^{N-1} a_i x^i : a_i \in \set{-1,0,1}}$. It has been used to improve the bootstrapping algorithm of FHE schemes, one of the main bottlenecks of the practicality of FHE applications.

Another (well-known) choice of $\vs$ is a discrete Gaussian with standard deviation $\sigma = 3.2$ \cite{OpenFHE, HElib, lattigo}, for which $\|\vs\|_\infty \le 6\sigma$ in the coefficient embedding with high probability at the concrete dimensions we test; asymptotically, a union bound over the $N$ coefficients gives $\|\vs\|_\infty = O(\sigma\sqrt{\log N})$ with overwhelming probability (concretely $S=\sigma\sqrt{2\ln(2N/\epsilon)}$ for failure probability $\epsilon$), matching the error bound used above. The quantity relevant to the per-index recovery of \Cref{sec: main} is this coefficient $\ell_\infty$ bound, so we take $S=1$ for ternary secrets and $S=6\sigma$ for Gaussian secrets, matching the values used in our experiments.

{%
\subsection{Public Keys and Coefficient-to-Slot Transform}\label{subsec: c2s}
Beyond the public encryption key $\pk=(\va,\vb=\va\vs+\ve)$, CKKS deployments often release \emph{rotation keys}. 
The rotation keys $\{\mathsf{rotk}_g\}_{g\in\Gamma}$, also called \emph{Galois keys}, are used to perform homomorphic ring automorphisms $\vp(x)\mapsto \vp(x^g)$ on ciphertexts. These automorphisms implement slot rotations, complex conjugation, and the Coefficient-to-Slot (C2S) transform.

We say that a CKKS deployment exposes \emph{{C2S}-capable public material} if a rotation-key set $\{{\sf rotk}_g\}_g$, sufficient to evaluate a {C2S} layer, is released. 
This is usually the surface of bootstrapping-enabled or coefficient/slot-conversion workloads, \emph{not} the default of every CKKS deployment; 
however, it represents the most common choices in secure outsourcing applications. 

One such sufficient set is the full automorphism group $\Gamma$ generated by $x\mapsto x^5$ together with the conjugation automorphism, which is the same type of rotation-key material exposed by Lattigo \cite{lattigo}, HEaaN \cite{heaan2022}, and OpenFHE \cite{OpenFHE} whenever the workload performs bootstrapping or coefficient/slot conversion. Although C2S is most commonly identified with CKKS bootstrapping \cite{BMTH21, CHKKS18, HS18}, the same rotation key set is also released whenever the workload invokes a homomorphic linear transform that requires the conjugation automorphism \cite{HS18} or any protocol that interconverts coefficient and slot representations. Deployments that release only a strict subset of $\Gamma$ insufficient to span a single C2S layer fall outside this assumption; a shallow-depth application sharing rotation keys for one fixed linear map is one such case.

\medskip\noindent\textbf{Normalized Coefficient-to-Slot Transform.}
The normalized C2S transform is a public homomorphic linear map on CKKS ciphertexts that 
{maps the coefficient vector of the encrypted polynomial into encrypted slot values.}
Concretely, given a ciphertext encrypting a polynomial $\vp=\sum_{j=0}^{N-1}p_j x^j$ at scale $\Delta$, we denote $\mathsf{C2S}_\Delta$ that returns a pair of ciphertexts $(\mathsf{ct}_R,\mathsf{ct}_I)$ at scale $\Delta$, each with $N/2$ slots. Overloading the decoding notation to ciphertexts, we write {$\mathsf{Decode}_\Delta(\mathsf{ct}):=\mathsf{Decode}(\langle\mathsf{ct},\sk\rangle\bmod q, \Delta)$} for the slot vector obtained by decrypting $\mathsf{ct}$ and decoding the resulting polynomial. With this convention,
\begin{gather*}
\mathsf{Decode}_\Delta(\mathsf{ct}_R)
\approx \Delta^{-1}\bigl(p_0,\,p_1,\,\ldots,\,p_{N/2-1}\bigr)\in\mathbb{R}^{N/2},\\
\mathsf{Decode}_\Delta(\mathsf{ct}_I)
\approx \Delta^{-1}\bigl(p_{N/2},\,p_{N/2+1},\,\ldots,\,p_{N-1}\bigr)\in\mathbb{R}^{N/2},
\end{gather*}
up to a bounded additive slot-domain noise $\ve_{\mathsf{C2S}}$; our attack requires $\|\ve_{\mathsf{C2S}}\|_\infty^{\mathrm{slot}}$ to satisfy the slot-margin of \Cref{thm: main}, which is then empirically validated in \Cref{sec:experiment} for the evaluated Lattigo regimes. In words, the input polynomial's coefficient vector, viewed at scale $\Delta$, is exposed as two real-valued slot vectors. The construction realizes this as the standard CKKS C2S layer~\cite{BMTH21, CHKKS18, HS18} followed by a conjugation-based real/imaginary split using the public conjugation rotation key.\footnote{Lattigo's \texttt{dft.RepackImagAsReal} format fuses both steps and returns the pair $(\mathsf{ct}_R,\mathsf{ct}_I)$ directly.} 
The transform consumes $\ell_{\mathsf{C2S}}$ multiplicative levels, where $\ell_{\mathsf{C2S}}\in\{3,4\}$ for typical bootstrap-quality decompositions used by~\cite{OpenFHE, heaan2022, lattigo}, and uses only the rotation keys.

We use $\mathsf{C2S}_\Delta$ as a public primitive: the attack presented in \Cref{sec: main} invokes $\mathsf{C2S}_\Delta$ exactly once and treats its slot-domain noise as a constant additive term in the final correctness analysis.
}

\section{Reaction Attack on CKKS and Its Variants}\label{sec: main}
This section presents the threat model and the required definitions, followed by our reaction attack with its correctness proof. We first fix the algebraic setting and secret bound, then define the observable slot-domain reaction oracle, and finally {show how the public rotation keys for C2S turn that
one bit into coordinate-wise recovery of $\vs$.}
Throughout, we let $\| \vs \|_\infty \le S$ in the coefficient embedding (i.e., $\max_j |s_j| \le S$), for some small integer $S >0$. This is standard for secret keys in the FHE context, using ternary secrets $S=1$.

\medskip\noindent\textbf{Two Metric Domains.}
Throughout the attack, the processed ciphertext transitions between two distinct representations: the coefficient domain and the slot domain. The secret bound $S$, the correctness requirement of \Cref{eqn: size_ckks}, and the modulus constraints in \Cref{thm: main} are all formulated in the coefficient domain, where the target of recovery is the integer coefficient $s_j$ of the secret key $\vs$ satisfying $\max_j |s_j| \le S$ (or holding with overwhelming probability over the choice of the secret key). By contrast, the reaction oracle operates strictly in the slot domain, where validity checks are evaluated via the $\|\cdot\|_\infty^{\mathrm{slot}}$ norm against a threshold $\tau$ (see \Cref{sec: reaction-oracle}). 

The public transformation C2S (see \Cref{subsec: c2s}) bridges these two domains, homomorphically relocating the secret key coefficients into corresponding slots. This allows the adversary to isolate a single coefficient and establish a direct operational link to the oracle's response bit (see \Cref{fig: attack-overview}).

{
\subsection{Threat Model}\label{sec: threat-model}
We consider a secure-outsourcing protocol where an honest client ({\sf C}) delegates homomorphic computation to a malicious server ({\sf S}). The client {\sf C} generates the keys, releases the public material, and applies an application-level admissibility predicate to the decoded slot values of any returned ciphertext before proceeding with the protocol. 
For each reaction query, the server {\sf S} receives only a single bit of leakage: whether the queried ciphertext is \textit{valid} or \textit{invalid}.

\medskip\noindent\textbf{Public Material.}
The client $\mathsf{C}$ publishes the public encryption key $\pk=(\va,\vb=\va\vs+\ve)$
and the set of public rotation keys $\{\mathsf{rotk}_g\}_g$ required to evaluate the application's homomorphic linear transforms, including the C2S keys when the workload utilizes it. 
The adversary may exploit only these honestly generated public values; affecting the key generation is strictly prohibited. This directly matches the standard ``honest key generation, malicious computation'' model established in prior literature~\cite{CCCM22, CCCM24,LMSV12}.

\medskip\noindent\textbf{Reaction Oracle.}
Upon decrypting and decoding a returned ciphertext, the client {\sf C} checks whether the magnitudes of the decoded slots lie within an admissible range bounded by a threshold $\tau$ (see \Cref{sec: reaction-oracle} for a typical instantiation). 
The resulting validation bit is observable to {\sf S} either directly (e.g., via an explicit abort signal) or indirectly (e.g., by observing whether the protocol continues or halts). 
We formally model this verification mechanism as a slot-domain reaction oracle $\mathcal{O}_{\mathrm{react}}$. 

\medskip\noindent\textbf{Real-World Reaction Leakage.}
The bit appears whenever the client's decoded slot vector is checked against an admissible range and the outcome affects an observable protocol step. Concrete examples include
(a) outsourced inference workflows where the client decodes class logits, treats a magnitude or range violation as a malformed response, and either re-issues the query or aborts, which the server observes as a retry/no-retry signal;
(b) protocols that use the decoded value as the input to a public follow-up computation, so that an obvious overflow leaks via the next-round ciphertext;
(c) deployments such as those discussed in~\cite{CCCM22, CCCM24, LM21, LMSS22} that already model an explicit \texttt{valid}~/~\texttt{invalid} response.

Note that while our attack systematically exploits this reaction bit to extract the key, it operates strictly at the protocol layer; it does not compromise the underlying RLWE hardness assumption or violate the IND-CPA security of the CKKS primitive itself.

\medskip\noindent\textbf{Adversarial Goal.}
For given public material and adaptive query access to the reaction oracle $\mathcal{O}_{\mathrm{react}}$, the adversary's (server's) goal is to fully recover the secret key $\vs$.

\medskip\noindent\textbf{Scope of the Attack.}
The reaction oracle arises under three deployment conditions taken together: (i) the client decrypts and decodes ciphertexts supplied by the evaluator; (ii) it applies an application-level admissibility check (e.g., a magnitude or range test) to the decoded slots; and (iii) the outcome of that check influences some behavior observable to the evaluator, such as a retry, an abort, an error, or a subsequent protocol message. Only when all three hold, the attack applies. 

In particular, two natural conditions place a deployment outside our scope. First, if every decryption-dependent behavior is made independent of the admissibility outcome (e.g., a standardized response regardless of the check), condition~(iii) fails and no reaction bit is exposed. Second, if the client decrypts only \emph{legitimate} ciphertexts, which is assumed by the IND-CPA$^\text{D}$ model,
removes the surface. Our threat model targets the common intermediate case in which arbitrary evaluator-supplied ciphertexts are decrypted and a decryption-dependent behavior is observable.

Since the attack is mounted with fresh ciphertexts for every query, the transcript does not exhibit strong query-to-query correlations, circumventing defenses that flag structured perturbations of a fixed target, though per-key monitoring still applies (\Cref{sec: discussion}). Consequently, relying on FHE alone is insufficient. Deployments must mitigate these reaction-based side channels, for example, by hardening client responses or limiting oracle access. 

We do not consider configurations where the server {\sf S} directly learns decrypted plaintext values, such as in standard chosen-ciphertext (CCA) or CPA$^\text{D}$ attack models. 
Furthermore, the server {\sf S} has no control over the key generation process, nor the client's internal randomness. 
Traditional physical-layer side channels, including timing analysis, power analysis, and fault-injection attacks, are similarly out of scope.

\subsection{Definition of Reaction Oracle}\label{sec: reaction-oracle}
We instantiate the reaction oracle directly in the \emph{decoded-slot} domain, mirroring the practical validation predicates a CKKS client applies post-decoding. 
The client decodes the returned ciphertext under its designated scale $\Delta$ and evaluates whether the resulting slot magnitudes conform to the deployment's expected range. 
We do not require the oracle to use any internal awareness or measurement of the underlying CKKS noise; only a simple range check is sufficient. 

Let $\tau > 0$ be an admissibility threshold defined by the deployment constraints. 
The default configuration in our framework is $\tau=1$, aligning with standard CKKS messages satisfying $\|\vm\|_\infty^{\mathrm{slot}} \le 1$. 
Given a queried ciphertext $\ct=(\va,\vb) \in \calR_q^2$, we denote the maximum decoded slot magnitude as $\mu(\ct)$, defined by:
\begin{align*}
    \mu(\mathsf{ct})
    :=& \|\vb-\va\vs \bmod q\|_\infty^{\mathrm{slot}} \\
    =& \|\mathsf{Decode}_\Delta(\vb-\va\vs \bmod q)\|_\infty\enspace.
\end{align*}
Using this metric, we formally define the \textit{slot-domain reaction oracle} as:
\begin{align*}
  \mathcal{O}_{\mathrm{react}}(\mathsf{ct}) = \begin{cases}
    \texttt{valid} & \text{if } \mu(\mathsf{ct})\le\tau,\\
    \texttt{invalid} & \text{otherwise.}
  \end{cases}
\end{align*}
The standard coefficient-domain correctness condition, $\|\vb-\va\vs \bmod q\|_\infty \le \Delta$, is employed across major FHE implementations~\cite{OpenFHE,lattigo, sealcrypto} to guarantee decryption correctness. 
This coefficient-level bound inherently implies a corresponding slot-domain bound, up to the operator norm of the canonical CKKS decoding map. 
In this work, however, we take the observable slot-domain predicate itself as the reaction oracle, precisely matching the client's system-level deployment scenario.
}

\medskip\noindent\textbf{Query Model.} 
The adversary is permitted to adaptively query $\mathcal{O}_{\mathrm{react}}$ with arbitrary ciphertexts under the target key. Crucially, the submitted ciphertexts are neither reused nor structurally unrealistic.

Instead, each query is independently randomized by adding a fresh or homomorphically evaluated encryption of zero (a construction we detail in \Cref{sec: attack-on}). As a result, the adversarial transcript never resubmits an identical ciphertext vector, nor does it present bit-by-bit deterministic modifications of a fixed target ciphertext. Prior to decryption, the query is a syntactically valid CKKS ciphertext, and is computationally indistinguishable from a fresh encryption under the RLWE assumption; hence, transcript-filtering or correlation-based defenses do not apply. However, per-key monitoring of query volume or rejection rate remains a valid detector (\Cref{sec: discussion}).

\subsection{Reaction Attack on CKKS}\label{sec: attack-on}

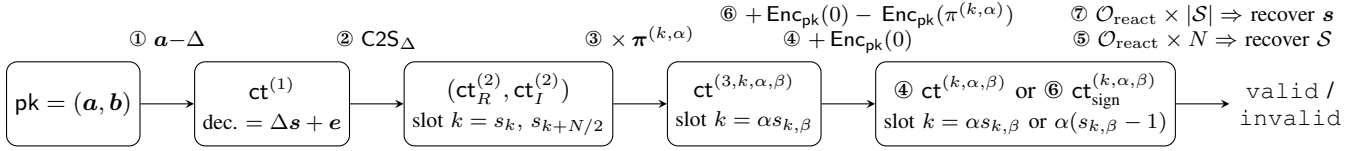
\begin{figure*}[t]
\centering
\resizebox{\textwidth}{!}{%
\begin{tikzpicture}[
  node distance=4mm and 7mm,
  box/.style={draw, rounded corners, align=center, minimum height=11mm, inner sep=3pt, font=\small},
  op/.style={font=\footnotesize, align=center, midway, above},
  >=stealth
]
\node[box] (pk) {$\pk=(\va,\vb)$};
\node[box, right=of pk] (c1) {$\mathsf{ct}^{(1)}$\\[1pt]\footnotesize dec.\ $=\Delta\vs+\ve$};
\node[box, right=of c1] (c2) {$(\mathsf{ct}_R^{(2)},\mathsf{ct}_I^{(2)})$\\[1pt]\footnotesize slot $k=s_k,\,s_{k+N/2}$};
\node[box, right=of c2] (c3) {$\mathsf{ct}^{(3,k,\alpha,\beta)}$\\[1pt]\footnotesize slot $k=\alpha s_{k,\beta}$};
\node[box, right=of c3] (sub) {\ding{175} $\mathsf{ct}^{(k,\alpha,\beta)}$ or \ding{177} $\mathsf{ct}_\textrm{sign}^{(k,\alpha,\beta)}$\\
[1pt]\footnotesize slot $k=\alpha s_{k,\beta}$ or $\alpha (s_{k,\beta}-1)$};
\node[right=7mm of sub, font=\small, align=center] (bit) {\texttt{valid}~/\\\texttt{invalid}};
\draw[->] (pk) -- node[op]{\ding{172} {$\va{-}\Delta$}\\ \\} (c1);
\draw[->] (c1) -- node[op]{\ding{173} {$\mathsf{C2S}_\Delta$}\\ \\} (c2);
\draw[->] (c2) -- node[op]{\ding{174} $\times\,\boldsymbol\pi^{(k,\alpha)}$\\ \\} (c3);
\draw[->] (c3) -- node[op]{~~~~~\ding{177} $+\,\mathsf{Enc}_\pk(0)-\,\mathsf{Enc}_\pk(\pi^{(k, \alpha)})$\\ \ding{175} $+\,\mathsf{Enc}_\pk(0)$\\ \\} (sub);
\draw[->] (sub) -- node[op]{\ding{178} $\mathcal{O}_{\mathrm{react}} \times |\calS| \Rightarrow$ recover $\vs$\\ 
\ding{176} $\mathcal{O}_{\mathrm{react}} \times N \Rightarrow$ recover $\calS$\\ \\} (bit);
\end{tikzpicture}%
}
\caption{
Overview of the Reaction Attack (\Cref{alg:on-attack}). 
For $0\le k < N/2$, $s_{k, \beta}$ denotes $s_k$ or $s_{k+N/2}$, depending on $\beta\in\{R, I\}$.
The attack proceeds as follows:
\textbf{(i) Offline (\ding{172}, \ding{173})}: 
\ding{172} Public-key shift, 
and \ding{173} homomorphic C2S locates a secret key coefficient into a target slots.
\textbf{(ii) Online Phase 1 (\ding{174}, \ding{175}, \ding{176})}: 
Support recovery in $N$ queries.
\ding{174} One-hot mask vector isolates a target slot, and
\ding{175} a fresh encryption re-randomizes it. 
\ding{176} Resulting 
query recovers the support $\mathcal{S} := \{j : s_j \neq 0\}$. 
\textbf{(iii) Online Phase 2 (\ding{174}, \ding{177}, \ding{178})}: 
Sign recovery in $|\calS|$ queries.
\ding{174} Isolation, 
\ding{177} re-randomization, and subtraction by an encryption of the mask, deviates the slot values, which 
\ding{178} reveals the sign when queried.
The full secret $\vs$ is then recovered. 
}
\label{fig: attack-overview}
\end{figure*}
{
The attack links each oracle query to a single coefficient of $\vs$ at a time, using $\pk$ and the rotation keys for $\mathsf{C2S}_\Delta$. 
\Cref{fig: attack-overview} sketches the pipeline and \Cref{alg:on-attack} gives it in full.
Note that, for better visibility, $k$, $\alpha$, and $\beta$ are omitted in the superscripts of $\ct^{(3)}$, $\ct$, and $\ct_\mathrm{sign}$ in \Cref{alg:on-attack}.

\medskip\noindent\textbf{Offline Step.}
The adversary first performs two public operations that are independent
of the oracle queries. Starting from the public encryption key
$\pk=(\va,\vb=\va\vs+\ve)$, it forms $\mathsf{ct}^{(1)}:=(\va-\Delta,\,\vb)$, which decrypts to $\Delta\vs+\ve$ under the target key, since $\vb-(\va-\Delta)\vs=\Delta\vs+\ve\pmod q$.\footnote{One may need to lower the level if $\pk$'s modulus is larger than the usual ciphertext modulus, typically via level down or modulus switching.}

The adversary then applies the public normalized C2S to obtain $(\mathsf{ct}_R^{(2)},\mathsf{ct}_I^{(2)})\leftarrow\mathsf{C2S}_\Delta(\mathsf{ct}^{(1)})$. 
This carries the secret key coefficients into slots. 
Concretely, when decrypted and decoded, the $k$-th slot of
$\mathsf{ct}_R^{(2)}$ contains $s_k$, while the $k$-th slot of
$\mathsf{ct}_I^{(2)}$ contains $s_{k+N/2}$, each within a noise of slot-norm $\ll\tau$.

\medskip\noindent\textbf{Online Step.}
In the online step, the adversary queries the ciphertexts obtained based on $\mathsf{ct}_R^{(2)}$ and $\mathsf{ct}_I^{(2)}$. 
We divide the queries into two types: (i) {Support queries}, which are employed to identify the support set $\mathcal{S} := \{j : s_j \neq 0\}$, and (ii) {Sign queries}, which are used to determine the sign associated with each index in $\mathcal{S}$.

\medskip\noindent\textit{Phase 1: Support query.}
A support query tests a single secret coefficient by means of a one-hot
plaintext mask.  Let \(\boldsymbol e_k\in\mathbb R^{N/2}\) denote the
\(k\)-th standard slot basis vector.  For an amplitude \(\alpha>0\),
define the one-hot plaintext mask
\[
    \boldsymbol\pi^{(k,\alpha)}
    := \mathsf{Encode}(\alpha\boldsymbol e_k,\Delta)\enspace,
\]
namely the CKKS plaintext encoding, at scale \(\Delta\), of the slot
vector with value \(\alpha\) in slot \(k\) and zero elsewhere.

For \(\beta\in\{R,I\}\), define the target secret coefficient by
\[
    s_{k,\beta}
    :=
    \begin{cases}
        s_k, & \text{if } \beta=R\enspace,\\
        s_{k+N/2}, & \text{if } \beta=I\enspace.
    \end{cases}
\]
The adversary computes
\[
    \mathsf{ct}^{(3,k,\alpha, \beta)}
    := \mathsf{ct}^{(2)}_\beta\cdot\boldsymbol\pi^{(k,\alpha)}\enspace,
\]
which contains a signal only in
the \(k\)-th slot. 
The signal is \(\alpha s_{k,\beta}\), up to bounded
slot-domain noise.
The adversary then re-randomizes the query by adding a fresh encryption
of zero (at the same level and scale):
\[
    \mathsf{ct}^{(k,\alpha,\beta)}
    := \mathsf{ct}^{(0)}+\mathsf{ct}^{(3,k,\alpha,\beta)},\qquad
    \mathsf{ct}^{(0)}\gets\mathsf{Enc}_{\pk}(0)\enspace.
\]
This fresh zero encryption does not change the decoded signal, but makes
the submitted ciphertext independent of previous queries. Therefore,
the underlying slot vector has the form
\[
    \alpha s_{k,\beta}\boldsymbol e_k+\widetilde\delta\enspace,
\]
where \(\widetilde\delta\) denotes the resulting slot-domain noise.
That is, the slot vector effectively constitutes a noisy one-hot vector, carrying $\alpha s_{k, \beta}$ in its $k$-th coordinate. 

We now fix an amplitude \(\alpha>0\) for the attack.  This value
is chosen so that the noise-only case remains below the reaction
threshold \(\tau\), whereas the signal produced by any nonzero ternary
coefficient exceeds \(\tau\).  The precise slot-margin and modulus
conditions ensuring this separation are stated in the correctness
analysis below.

The reaction oracle then distinguishes whether
the targeted coefficient is zero:
\[
    \calO_{\mathrm{react}}
    \bigl(\mathsf{ct}^{(k,\alpha,\beta)}\bigr)
    =
    \begin{cases}
        \mathsf{valid}, & \text{if } s_{k,\beta}=0\enspace,\\
        \mathsf{invalid}, & \text{if } s_{k,\beta}\neq0\enspace.
    \end{cases}
\]
Thus, one oracle query determines whether the coefficient selected by
\((k,\beta)\) lies in the support $\calS$. 
Repeating this for
all \(k\in\{0,\ldots,N/2-1\}\) and \(\beta\in\{R,I\}\) recovers
$\mathcal S=\{j:s_j\neq0\}$ using exactly \(N\) reaction queries.

\medskip\noindent\textit{Phase 2: Sign query.}
A second pass over \(\mathcal S\) resolves the signs.  For
\(j\in\mathcal S\), choose \((k,\beta)\) so that \(s_{k,\beta}=s_j\).
The adversary first forms a fresh support-type query ciphertext
\(\mathsf{ct}^{(k,\alpha,\beta)}\) as above, then subtracts an
encryption of the one-hot mask:
\[
    \mathsf{ct}_{\mathrm{sign}}^{(k,\alpha,\beta)}
    :=
    \mathsf{ct}^{(k,\alpha,\beta)}
    -\mathsf{ct}_\mathrm{off}^{(k,\alpha)}\enspace,
\]
where $\mathsf{ct}_\mathrm{off}^{(k,\alpha)}
\gets \mathsf{Enc}_{\pk}\big(\boldsymbol\pi^{(k,\alpha)}\big)$.
Then the underlying target slot has signal
\[
    \alpha(s_{k,\beta}-1) \enspace,
\]
up to a bounded slot-domain noise. 
The signal is \(0\) when
\(s_{k,\beta}=+1\), and \(-2\alpha\) when \(s_{k,\beta}=-1\); hence, 
the oracle accepts exactly the
positive coefficients and rejects the negative ones. 
Thus the sign of
each nonzero coefficient is recovered with one additional query, and the
total query count becomes
\[
    N+|\mathcal S|\le 2N\enspace.
\]
All per-query multiplications are by plaintext masks, so 
no relinearization key is required.\footnote{In
Lattigo, \texttt{MulRelinNew} on a plaintext operand does not use
\(\mathsf{rlk}\).}

}

\begin{algorithm}[t]
\caption{$O(N)$ Reaction Attack on CKKS.}
\label{alg:on-attack}
\SetAlgoLined
\SetAlgoNoEnd
\DontPrintSemicolon
\KwIn{Slot-domain reaction oracle
$\mathcal{O}_{\mathrm{react}}$ with threshold $\tau$,
public material $(\pk=(\va, \vb)$, $\{\mathsf{rotk}_g\})$, parameters $(N,q_0,\Delta)$,
and a constant $\alpha$ satisfying the slot-margin and modulus
conditions of \Cref{thm: main}.}
\KwOut{Ternary secret $\vs=(s_0,\ldots,s_{N-1})$.}

$\mathsf{ct}^{(1)} \gets (\va-\Delta,\,\vb)$\tcp*{Offline shift}
$(\mathsf{ct}_R^{(2)},\mathsf{ct}_I^{(2)})
 \gets \mathsf{C2S}_\Delta(\mathsf{ct}^{(1)})$\tcp*{Offline C2S}
$\mathcal{S}\gets\emptyset$, $\vs\gets \mathbf 0$\;

\tcc{Phase~1: Support recovery, one query per coefficient.}
\ForEach{$\beta\in\{R,I\}$}{
    \For{$k=0$ \KwTo $N/2-1$}{
        $j \gets k$ if $\beta=R$, else $k+N/2$\;
        $\mathsf{ct}^{(3)}
        \gets \mathsf{ct}_\beta^{(2)}
        \cdot \boldsymbol\pi^{(k,\alpha)}$\;
        Sample fresh $\mathsf{ct}^{(0)}\gets\mathsf{Enc}_{\pk}(0)$\;
        $\mathsf{ct} \gets \mathsf{ct}^{(0)}+\mathsf{ct}^{(3)}$\;
        \If{$\mathcal{O}_{\mathrm{react}}(\mathsf{ct})
        = \texttt{invalid}$}{
            $\mathcal{S}\gets \mathcal{S}\cup\{j\}$\;
        }
    }
}

\tcc{Phase~2: Sign recovery, one query per nonzero index.}
\ForEach{$j\in\mathcal{S}$}{
    $\beta \gets R$ if $j<N/2$ else $I$\;
    $k\gets j$ if $\beta=R$ else $j-N/2$\;
    $\mathsf{ct}^{(3)}
    \gets \mathsf{ct}_\beta^{(2)}
    \cdot \boldsymbol\pi^{(k,\alpha)}$\;
    $\mathsf{ct}_{\mathrm{off}}^{(k,\alpha)}
    \gets \mathsf{Enc}_{\pk}(\boldsymbol\pi^{(k,\alpha)})$\;
    Sample fresh $\mathsf{ct}^{(0)}\gets\mathsf{Enc}_{\pk}(0)$\;
    $\mathsf{ct}_\mathrm{sign}
    \gets \mathsf{ct}^{(0)}
    +\mathsf{ct}^{(3)}
    -\mathsf{ct}_{\mathrm{off}}^{(k,\alpha)}$\;
    \If{$\mathcal{O}_{\mathrm{react}}(\mathsf{ct}_\mathrm{sign})
    = \texttt{valid}$}{
        $s_j \gets +1$\;
    }\Else{
        $s_j \gets -1$\;
    }
}
\Return{$\vs$}
\end{algorithm}

{
\medskip\noindent\textbf{Mask Amplitude.}
The offline and per-query noise terms combine as follows.  Phase~1 and Phase~2 queries decode respectively to
\[
    \alpha s_{k,\beta}\boldsymbol e_k+\widetilde\delta
    \quad\text{and}\quad
    \alpha(s_{k,\beta}-1)\boldsymbol e_k+\widetilde\delta' \enspace.
\]
Write
\[
    \nu_\alpha
    :=
    B_{\mathsf{shift}}^{\mathrm{coeff}}/\Delta
    +B_{\mathsf{C2S}}^{\mathrm{slot}}\enspace,
\]
for the slot-domain noise carried through \(\mathsf{C2S}_\Delta\) per
unit mask amplitude.  Then the noise terms are bounded by
\begin{align*}
\|\widetilde\delta\|_\infty^{\mathrm{slot}}
&\le
\eta_{\sup}
:=
\alpha\nu_\alpha
+B_{\mathsf{pmult}}^{\mathrm{slot}}
+B_{\mathsf{fresh}}^{\mathrm{slot}}\enspace,\\
\|\widetilde\delta'\|_\infty^{\mathrm{slot}}
&\le
\eta_{\mathrm{sign}}
:=
\eta_{\sup}
+B_{\mathsf{off}}^{\mathrm{slot}}\enspace,
\end{align*}
where \(B_{\mathsf{pmult}}^{\mathrm{slot}}\),
\(B_{\mathsf{fresh}}^{\mathrm{slot}}\), and
\(B_{\mathsf{off}}^{\mathrm{slot}}\) are the noise bounds introduced by the
plaintext masking, the fresh encryption of zero, and the Phase-2 offset,
respectively.

The role of \(\alpha\) is to separate zero and nonzero ternary
coefficients with respect to the reaction threshold \(\tau\).  When
\(s_{k,\beta}=0\), the Phase-1 query contains only noise; when
\(s_{k,\beta}\neq0\), the targeted slot contains a signal of magnitude
approximately \(\alpha\).  Since the mask scales both the signal
\(\alpha s_{k,\beta}\) and the inherited C2S noise
\(\alpha\nu_\alpha\), changing \(\alpha\) does not improve their
ratio.  It only places the signal and noise relative to the fixed
threshold \(\tau\).

Ignoring the smaller additive terms
\(B_{\mathsf{pmult}}^{\mathrm{slot}}\),
\(B_{\mathsf{fresh}}^{\mathrm{slot}}\), and
\(B_{\mathsf{off}}^{\mathrm{slot}}\), Phase~1 separates zero from
nonzero coefficients whenever
\[
    \tau < \alpha < \tau/\nu_\alpha \enspace.
\]
The full slot-margin inequalities, including these additive terms and
the Phase-2 sign separation, are stated in \Cref{thm: main}.  This
admissible window is independent of \(N\), since both the slot signal
\(\alpha\) and the threshold \(\tau\) are independent of \(N\).  In
our experimental regimes, \(\nu_\alpha\ll1\) holds with a large margin
(see \Cref{sec:experiment}), so we fix
$\alpha$ to $4$.
}

\subsection{Correctness and Query Complexity}

\begin{thm}\label{thm: main}
Consider a CKKS deployment with parameters $(N,q_0,\Delta)$, level-zero
modulus $q_0$, ternary secret $\vs\in\{-1,0,1\}^N$, public material
$(\pk,\{\mathsf{rotk}_g\})$ sufficient for $\mathsf{C2S}_\Delta$, and
reaction oracle $\mathcal{O}_{\mathrm{react}}$ with threshold
$\tau\ge1$. Let $\eta_{\sup}$ and
$\eta_{\mathrm{sign}}=\eta_{\sup}+B_{\mathsf{off}}^{\mathrm{slot}}$
 the noise bounds of the decoded slots of Phase~1 and Phase~2 queries,
respectively. Suppose an amplitude $\alpha>0$ satisfies:

\emph{(i) Slot-margin.}
\begin{align}
\eta_{\sup}&<\tau\enspace,
&\alpha-\eta_{\sup}&>\tau\enspace, \label{eq: alpha-margin}\\
\eta_{\mathrm{sign}}&<\tau\enspace,
&2\alpha-\eta_{\mathrm{sign}}&>\tau\enspace. \notag
\end{align}

\emph{(ii) No modular wrap.}
\begin{equation}\label{eq: modulus}
2\alpha\,\kappa_{\max}+B_{\mathrm{coeff}}<q_0/2\enspace,
\end{equation}
where $\kappa_{\max}\le\Delta/\sqrt{N/2}$ is the largest
coefficient-domain magnitude produced by encoding a unit-magnitude slot
and $B_{\mathrm{coeff}}$ bounds the coefficient-domain non-signal terms.

Then \Cref{alg:on-attack} recovers $\vs$ in at most $2N$
reaction-oracle queries and $\poly(N,\log q)$ time, except with
negligible probability.
\end{thm}

{
\begin{proof}[Proof of \Cref{thm: main}]
Fix a pair $(k,\beta)$ with $k\in\{0,\ldots,N/2-1\}$ and
$\beta\in\{R,I\}$, and write $s_{k,\beta}$ for the targeted coefficient:
$s_{k,R}=s_k$ and $s_{k,I}=s_{k+N/2}$.

\medskip\noindent\emph{Decoded slot values.} 
The construction of \Cref{sec: attack-on} shows that the Phase~1 query ciphertext decodes, in the slot domain, to $\alpha s_{k,\beta}\,\boldsymbol{e}_k+\widetilde\delta$ with $\|\widetilde\delta\|_\infty^{\mathrm{slot}}\le\eta_{\sup}$, and the Phase~2 query ciphertext to $\alpha(s_{k,\beta}-1)\,\boldsymbol{e}_k+\widetilde\delta'$ with $\|\widetilde\delta'\|_\infty^{\mathrm{slot}}\le\eta_{\mathrm{sign}}$; the additional $B_{\mathsf{off}}^{\mathrm{slot}}$ in $\eta_{\mathrm{sign}}$ accounts for the fresh encryption of the offset $\boldsymbol\pi^{(k,\alpha)}$ subtracted in the Phase-2 step of \Cref{alg:on-attack}.\footnote{If the offset is implemented by plaintext subtraction $\mathsf{ct}-\boldsymbol\pi^{(k,\alpha)}$ in place of $\mathsf{ct}-\mathsf{Enc}_{\pk}(\boldsymbol\pi^{(k,\alpha)})$, the term $B_{\mathsf{off}}^{\mathrm{slot}}$ vanishes; we use the encrypted offset to keep the submitted ciphertext distribution close to a fresh CKKS ciphertext.}

\medskip\noindent\emph{No modular wrap (Condition~(ii)).} 
The decoded values read the level-zero decryption residue through the centered reduction modulo $q_0$, which preserves the signal only if no coordinate leaves $[-q_0/2,q_0/2)$. 
The largest coefficient-domain magnitude is the Phase~2 worst case $s_{k,\beta}=-1$, whose encoded slot vector $\alpha(s_{k,\beta}-1)\boldsymbol{e}_k$ has $\ell_\infty$-norm $2\alpha\kappa_{\max}$; together with the non-signal terms bounded by $B_{\mathrm{coeff}}$ the total is below $q_0/2$ by \Cref{eq: modulus}. Hence no coordinate wraps and the slot identities above hold.

\medskip\noindent\emph{Phase~1 (support).} If $s_{k,\beta}=0$ the decoded slot vector equals $\widetilde\delta$, of norm $\le\eta_{\sup}<\tau$ by \Cref{eq: alpha-margin}, so the oracle accepts. If $s_{k,\beta}\ne0$ then $|s_{k,\beta}|=1$, so slot~$k$ has magnitude $\ge\alpha-\eta_{\sup}>\tau$ and the oracle rejects. One query at each $(k,\beta)$ with $k\in[N/2]$, $\beta\in\{R,I\}$, therefore decides whether the addressed coefficient is zero, and the $N$ such queries identify the support $\mathcal{S}=\{j:s_j\ne0\}$.

\medskip\noindent\emph{Phase~2 (sign).} For each $j\in\mathcal{S}$ the targeted slot holds $\alpha(s_{k,\beta}-1)+\widetilde\delta'_k$ with $s_{k,\beta}-1\in\{0,-2\}$. By \Cref{eq: alpha-margin}, $s_{k,\beta}=+1$ gives magnitude $\le\eta_{\mathrm{sign}}<\tau$ (accept) and $s_{k,\beta}=-1$ gives magnitude $\ge 2\alpha-\eta_{\mathrm{sign}}>\tau$ (reject), so one query per $j\in\mathcal{S}$ recovers its sign.

\medskip\noindent\emph{Query count and running time.} Phase~1 issues $N$ queries and Phase~2 issues $|\mathcal{S}|\le N$, for at most $2N$ in total; each query is assembled from public CKKS operations, so the algorithm runs in $\poly(N,\log q)$ time.

\medskip\noindent\emph{Success probability.} The key-generation and evaluation-key randomness is drawn once; conditioned on it, the $\mathsf{C2S}_\Delta$ and key-switching noise terms entering $B_{\mathsf{C2S}}^{\mathrm{slot}}$ and $B_{\mathrm{coeff}}$ meet their stated bounds except with negligible probability. Conditioned on such a key, each query's remaining noise depends only on that query's fresh-encryption randomness and stays within $\eta_{\sup}$, $\eta_{\mathrm{sign}}$ except with negligible probability; a union bound over the $2N$ queries then makes every oracle answer correct, so $\vs$ is recovered except with negligible probability.
\end{proof}
}

{%
\Cref{thm: main} gives $2N$ queries; we close by showing this is essentially optimal for any one-bit reaction oracle.
\begin{thm}\label{thm: lower-bound}
In the oracle-only model where the public material is treated as information-theoretically independent of $\vs$, any algorithm that recovers every uniform $\vs\in\{-1,0,1\}^N$ from $\mathcal{O}_{\mathrm{react}}$ makes at least $\lceil N\log_2 3\rceil$ queries in the worst case. For the sparse-ternary distribution with $\Pr[s_i\ne 0]=p$, where nonzero coefficients have uniform signs, the average-case lower
bound on the number of binary oracle answers is
\[
H_p(\vs) = N\bigl(p\log_2(2/p)+(1-p)\log_2\tfrac{1}{1-p}\bigr)\enspace.
\]
\end{thm}
\begin{proof}
A transcript of $T$ adaptive binary oracle answers takes at most $2^T$ distinct values. A zero-error algorithm must induce distinct transcripts on distinct secrets, so $2^T\ge 3^N$ for uniform ternary secrets, giving $T\ge N\log_2 3$ in the worst case. For the Bernoulli-sparse ternary distribution,  each coefficient is zero
with probability \(1-p\), and is \(+1\) or \(-1\) with probability
\(p/2\) each. %
Shannon's source-coding bound gives the average-case lower bound $\mathbb{E}[T]\ge H_p(\vs)$ with $H_p$ defined as above.
\end{proof}

Together with \Cref{thm: main}, the query complexity of \Cref{alg:on-attack} remains within a factor of $2/\log_2 3 \approx 1.262$ of this information-theoretic lower bound. Note that this abstraction inherently ignores the computational information that the RLWE public key $\pk$ carries regarding $\vs$; thus, we claim near-optimality strictly with respect to the number of one-bit oracle responses, rather than against an unbounded adversary that additionally exploits $\pk$. Concretely, the default Lattigo~\cite{lattigo} configuration with a uniform ternary secrets ($p=2/3$) has $H_p(\vs) \approx 1.585N$. In this regime, the empirical query count of $\approx 1.667N$ observed in \Cref{sec:experiment} is within a factor of $1.05$, demonstrating the near-optimality of our attack in practice.
}

\begin{rmk}\label{rmk: bounded-secret}
The same idea extends to secrets with coefficients in
\(\{-S,\ldots,S\}\).  Instead of testing only whether a coefficient is
zero, the adversary can perform interval tests.  To test whether
\(s_j\) lies within distance \(w\tau\) of a center \(c\), it uses a
one-hot mask of amplitude \(\alpha=1/w\) and subtracts an encrypted
one-hot offset encoding \(c\) in the targeted slot.  The decoded target
slot then has the form
\[
    (s_j-c)/w+\widetilde\delta\enspace.
\]
Under the corresponding slot-margin conditions, the reaction oracle
tests whether \(|s_j-c|\le w\tau\).  Applying bisection over
\(\{-S,\ldots,S\}\) recovers each coefficient in
\(\lceil\log_2(2S+1)\rceil\) queries, hence \(O(N\log S)\) queries in
total.
\end{rmk}

\begin{rmk}[Composition with an honest output]\label{rmk: composition}
The Phase-1 and Phase-2 query ciphertexts may instead be added to a ciphertext obtained from an honest evaluation of a circuit that the secret key owner expects. The only change is that the isolated target slot now reads $m_k+\alpha s_{k,\beta}$ rather than $\alpha s_{k,\beta}$; it stays admissible when $s_{k,\beta}=0$ and overflows otherwise, so the reaction still isolates the target while riding on a legitimate result.
\end{rmk}

\subsection{Sparse Secrets and CKKS Variants}
\label{sec: applicability}

The attack described in \Cref{alg:on-attack} can be generalized in two directions: 1) we can further reduce the query complexity for CKKS with sparse
ternary secrets, and 2) we can target CKKS variants whenever they provide a public, normalized \(\mathsf{C2S}_\Delta\) layer that meets the necessary noise margins. 

\subsubsection{Sublinear Queries for Sparse Secrets}
The baseline $2N$-query bound allocates one oracle query per coordinate to recover the support $\calS$. 
However, when the secret key is sparse, almost all of its uncertainty lies in the support, which can be optimized with far fewer queries using a combinatorial \emph{Group Testing} strategy. 
Since the reaction oracle thresholds the {maximum} slot-norm magnitude, 
the individual one-hot mask ($\pi^{(k, \alpha)}$) can be naturally generalized to a \emph{subset mask}:
\[
\boldsymbol{\pi}^{(T,\alpha)} := \mathsf{Encode}\bigg( \alpha\sum_{k\in T}\mathbf{e}_k, \Delta \bigg)\enspace,
\]
over any set of slots $T$. 

The resulting reaction query triggers an \texttt{invalid} response if and only if at least one coordinate within $T$ intersects with the secret support. 
Crucially, because the oracle evaluates the maximum magnitude across distinct slots, the individual $\pm1$ signs of the secret coefficients occupy isolated positions and cannot cancel each other out. 
Each query, therefore, effectively functions as a single group test, asking \emph{``Does the subset $T$ intersect the secret support $\mathcal{S}$?''} 
Consequently, an adaptive bisection procedure can identify all $h = |\mathcal{S}|$ non-zero coordinates using only $O(h\log_2(N/h))$ queries.

Once the support $\mathcal{S}$ is fully recovered, the corresponding signs can be extracted directly from the public key without any further oracle queries. 
Restricting the decryption relation $\vb = \va\vs + \ve \pmod q$ to the support $\mathcal{S}$ leaves an over-determined, $h$-dimensional LWE instance ($N$ samples in $h \ll N$ unknowns), where the short secret is the sign vector $\vs\vert_{\mathcal{S}} \in \{-1, +1\}^h$. 
The concrete hardness of this reduced instance is evaluated empirically in \Cref{sec:experiment}.

The full algorithm is specified in \Cref{alg:algo_sparse}, where the recursive binary-splitting procedure is formalized as $\mathsf{BisecGroupTest}$ in \Cref{alg:subroutine_hwang}. The lattice reduction ${\sf LatticeReduction}_q(\mathbf{A}, \vb)$ finds a short vector $\vs$ satisfying $\mathbf{A}\vs \approx \vb \bmod q$. 
Since the LWE dimension collapses from $N$ to the sparse Hamming weight $h \ll N$ while $q$ remains large, this reduced instance can be solved efficiently using standard lattice reduction algorithms (e.g., LLL or BKZ).
Note that an adversary can construct the modular relation from any $\ct$, $\pk$, or ${\sf rotk}$; selecting a key with the largest available modulus $q = PQ$ will make the reduction cost optimal.

\begin{algorithm}[t]
\caption{Reaction Attack on Sparse-secret CKKS.}
\label{alg:algo_sparse}
\SetAlgoLined
\SetAlgoNoEnd
\DontPrintSemicolon
\KwIn{Reaction oracle $\mathcal{O}_{\mathrm{react}}$ with threshold $\tau$, public material $(\pk=(\va, \vb),\{\mathsf{rotk}_g\})$, modulus $q$ of \pk, parameters $(N,q_0,\Delta)$, and amplitude $\alpha$ satisfying \Cref{thm: sparse}.}
\KwOut{Sparse ternary secret $\vs=(s_0,\ldots,s_{N-1})$.}

$\mathsf{ct}^{(1)} \gets (\va-\Delta,\,\vb)$\tcp*{Offline shift}
$(\mathsf{ct}_R^{(2)},\mathsf{ct}_I^{(2)}) \gets \mathsf{C2S}_\Delta(\mathsf{ct}^{(1)})$\tcp*{Offline C2S}
$\mathcal{S}\gets\emptyset$, $\vs\gets \mathbf 0$\;

\tcc{Phase~1: Support recovery.}
\ForEach{$\beta\in\{R,I\}$}{
    $\mathcal{S}_\beta \gets \mathsf{BisecGroupTest}\left([\frac{N}{2}-1], \, \mathsf{ct}_\beta^{(2)}\right)$\;
    $\mathcal{S} \gets \mathcal{S} \cup \mathcal{S}_\beta$\;
}

\tcc{Phase~2: Sign recovery.}
Construct $\mathbf{M}_{\va} \in \mathbb{Z}_q^{N \times N}$ from $\va$\;
$\vs\vert_{\mathcal{S}} \gets \mathsf{LatticeReduction}_{q}\left(\mathbf{M}_{\va}\vert_{\mathcal{S}}, \, \vb\right)$\;
$\vs\vert_{\calS^c} \gets 0$\;
\Return{$\vs$}
\end{algorithm}

\begin{algorithm}[t]
\caption{Subroutine $\mathsf{BisecGroupTest}$.}
\label{alg:subroutine_hwang}
\SetAlgoLined
\SetAlgoNoEnd
\DontPrintSemicolon
\KwIn{Target $T \subseteq [\frac{N}{2}-1]$, and a ciphertext $\mathsf{ct}^{(2)}$.}
\KwOut{Local support indices $\mathcal{S}_{\mathrm{local}}$.}
\KwData{Global Parameters: $\mathcal{O}_{\mathrm{react}}$, $\pk$, $\Delta$, $\alpha$}

$\mathcal{S}_{\mathrm{local}} \gets \emptyset$\;
$\boldsymbol{\pi}^{(T,\alpha)} \gets \mathsf{Encode}\left( \alpha\sum_{k\in T}\mathbf{e}_k, \Delta \right)$ %
Sample fresh $\mathsf{ct}^{(0)}\gets\mathsf{Enc}_{\pk}(0)$\;
$\mathsf{ct} \gets \mathsf{ct}^{(0)}+\mathsf{ct}^{(2)} \cdot \boldsymbol{\pi}^{(T,\alpha)}$\;

\BlankLine
\If{$\mathcal{O}_{\mathrm{react}}(\mathsf{ct}) == \texttt{invalid}$}{
    \If{$|T| == 1$}{
        $k \gets \text{the single element in } T$\;
        \If{$\mathsf{ct}^{(2)} == \mathsf{ct}_R^{(2)}$}{
            $\mathcal{S}_{\mathrm{local}} \gets \mathcal{S}_{\mathrm{local}} \cup \{k\}$\;
        }\Else{
            $\mathcal{S}_{\mathrm{local}} \gets \mathcal{S}_{\mathrm{local}} \cup \{k+ \frac{N}{2}\}$\;
        }
    }\Else{
        Partition $T:=T_1\sqcup T_2$, s.t. $|T_1| = \lfloor |T|/2 \rfloor$\;
        $\mathcal{S}_{\mathrm{local}} \gets \mathcal{S}_{\mathrm{local}} \cup \mathsf{BisecGroupTest}\left(T_1, \mathsf{ct}^{(2)}\right)$\;
        $\mathcal{S}_{\mathrm{local}} \gets \mathcal{S}_{\mathrm{local}} \cup \mathsf{BisecGroupTest}\left(T_2, \mathsf{ct}^{(2)}\right)$\;
    }
}
\Return{$\mathcal{S}_{\mathrm{local}}$}
\end{algorithm}

\Cref{thm: sparse} then proves the correctness of \Cref{alg:algo_sparse}.

\begin{thm}\label{thm: sparse}
Let $\vs$ be a ternary secret with Hamming weight $h < N$, and assume the Phase-1 slot-margin of \Cref{thm: main} together with the multi-slot modulus condition $\alpha\kappa_{\max}h+B_{\mathrm{coeff}}<q_0/2$. 
Then, the group testing in \Cref{alg:algo_sparse} recovers the support $\mathcal{S}$ using $O(h\log_2(N/h))$ reaction queries. The sign vector $\vs\vert_{\mathcal{S}}$ is subsequently determined by solving the restricted $h$-dimensional LWE instance, requiring no additional oracle queries.
\end{thm}

\begin{proof}[Proof of \Cref{thm: sparse}]
Recall the Phase-1 slot-margin condition $\eta_{\sup} < \tau < \alpha - \eta_{\sup}$ from \Cref{eq: alpha-margin}.

\medskip\noindent\emph{A subset query for group testing.}
Fix $\beta \in \{R, I\}$ and a slot subset $T \subseteq [N/2]$. 
We construct the subset mask $\boldsymbol{\pi}^{(T,\alpha)}$
and compute $\mathsf{ct} = \mathsf{ct}^{(0)} + \mathsf{ct}_\beta^{(2)} \cdot \boldsymbol{\pi}^{(T,\alpha)}$, where $\mathsf{ct}^{(0)} \gets \mathsf{Enc}_\pk(0)$. 
By the coordinate-wise linearity, the slot vector becomes
\[
    \alpha\sum_{k\in T} s_{j_k}\mathbf{e}_k + \widetilde{\boldsymbol{\delta}} \qquad \text{for some } \|\widetilde{\boldsymbol{\delta}}\|_\infty^{\mathrm{slot}} \le \eta_{\sup}\enspace,
\]
where $j_k = k$ for $\beta = R$ and $j_k= k + N/2$ for $\beta = I$. 
The noise bounds $B_{\mathsf{pmult}}^{\mathrm{slot}}$ and $B_{\mathsf{fresh}}^{\mathrm{slot}}$ scale with $\alpha$ and the ring parameters, remaining invariant to the cardinality $|T|$. 

Thus, if $T \cap \mathcal{S} = \emptyset$, then every $s_{j_k} = 0$, bounding the decoded magnitude by $\eta_{\sup} < \tau$ (\texttt{valid}). 
Else if $T \cap \mathcal{S} \neq \emptyset$, there exists at least one slot $k \in T$ with $s_{j_k} = \pm1$, yielding a magnitude of at least $\alpha - \eta_{\sup} > \tau$, triggering an \texttt{invalid} response. 
The query output thus perfectly instantiates the group-testing predicate $\mathbf{1}_{[\,T \,\cap \,\mathcal{S} \,\neq\, \emptyset\,]}$.

For the coefficient-domain modulus constraint, linearity and the triangle inequality bound the $\ell_\infty$-norm of the multi-slot signal by $\alpha\kappa_{\max} \cdot |T \cap \mathcal{S}| \le \alpha\kappa_{\max}h$. Modular wrap-around is securely avoided under the condition $\alpha\kappa_{\max}h + B_{\mathrm{coeff}} < q_0/2$.

\medskip\noindent\emph{Support recovery via repeated group testing.}
We run group testing on each half, independently using $\mathsf{ct}_\beta^{(2)}$. Let $h_R = |\calS \cap [0, M)|$ and $h_I = |\calS \cap [M, N)|$, where $h_R + h_I = h$ and $M = N/2$. 
Hwang's generalized binary-splitting algorithm~\cite{DuHwang00,Hwang72} isolates all $h_\beta$ non-zero indices within $\le O(h_\beta \log_2(M / h_\beta))$ adaptive queries, which is order-optimal against the information-theoretic bound $\lceil \log_2 \binom{M}{h_\beta} \rceil$ for $\beta \in \{R, I\}$. 
By the concavity of $t \mapsto t \log_2(M/t)$, the total query complexity is bounded by:
\[
\sum_{\beta \in \{R, I\}} h_\beta \log_2(M / h_\beta) \le h \log_2(2M / h) = h \log_2(N/h)\enspace.
\]
Hence, the entire support $\mathcal{S}$ is recovered in $O(h\log_2(N/h))$ reaction queries.

\smallskip
\noindent\emph{Sign extraction via lattice reduction.}
When represented in coefficient vectors, the public key elements satisfy $\vb = \mathbf{M}_{\va}\vs + \ve \pmod q$, where $\mathbf{M}_{\va} \in \mathbb{Z}_q^{N \times N}$ is the negacyclic matrix of $\va$. Since $\vs$ is supported strictly on $\mathcal{S}$, the relation restricts to the $h=|\calS|$ columns of $\mathbf{M_a}$:
\[
\vb = \mathbf{M_a}\vert_{\mathcal{S}} \cdot \vs\vert_{\mathcal{S}} + \ve \pmod q \enspace,
\]
for $\mathbf{M_a}\vert_{\mathcal{S}} \in \mathbb{Z}_q^{N \times h}$ and $\vs\vert_{\mathcal{S}} \in \{-1, +1\}^h$. 
This over-determined $h$-dimensional LWE instance can be solved via standard reductions, with no additional oracle access.
\end{proof}

Note that given $\kappa_{\max} \le \Delta / \sqrt{N/2}$, the multi-slot modulus constraint is highly permissive and easily satisfied in typical sparse secret deployments where $h \ll N$. Furthermore, while \Cref{thm: lower-bound} establishes a rigid reaction-query lower bound based purely on oracle feedback, this variant actively integrates $\pk$. Consequently, the total query complexity may drop below the $H_p(\vs)$ reaction-query floor whenever the $h$-dimensional LWE instance is efficiently solvable.

\subsubsection{Conditional Applicability to CKKS Variants}
Our attack framework extends beyond plain CKKS to architectural variants such as bootstrappable~\cite{BMTH21,CHKKS18, HS18}, rescaling-free~\cite{KDE+23, MHWW24}, and discrete-message~\cite{CHKS25, GV25,LW24} extensions. 
Any variant exposing a public, normalized $\mathsf{C2S}_\Delta$ layer with slot-wise plaintext multiplication is inherently vulnerable; its precise viability depends on the specific variant's embedding error, scale management, and decryption noise, as analyzed in Corollary
~\ref{cor: variants}. 

\begin{cor}\label{cor: variants}
Let $\Pi'$ be a CKKS variant whose decryption equation satisfies $\langle\mathsf{ct},\sk\rangle \bmod q = \vm + \mathbf{e} + \mathbf{e}_{\mathrm{aux}}$ with $\|\mathbf{e}_{\mathrm{aux}}\|_\infty \le B_{\mathsf{aux}}$. Define the augmented noise margins:
\begin{align*}
\eta_{\sup}'    &:= \alpha\bigl((B_{\mathsf{shift}}^{\mathrm{coeff}} + B_{\mathsf{aux}})/\Delta + B_{\mathsf{C2S}}^{\mathrm{slot}}\bigr) + B_{\mathsf{pmult}}^{\mathrm{slot}} + B_{\mathsf{fresh}}^{\mathrm{slot}}\enspace, \\
\eta_{\mathrm{sign}}' &:= \eta_{\sup}' + B_{\mathsf{off}}^{\mathrm{slot}}\enspace.
\end{align*}
If $\Pi'$ satisfies the augmented slot-margin conditions:
\[
\eta_{\mathrm{sign}}' < \tau\enspace, \quad \alpha - \eta_{\sup}' > \tau\enspace, \quad \text{and }\; 2\alpha - \eta_{\mathrm{sign}}' > \tau\enspace,
\]
then \Cref{alg:on-attack} still extracts the secret key $\vs$ within $2N$ queries.
\end{cor}

Note that variants utilizing IND-CPA\({}^{\mathrm{D}}\)-resistant noise flooding~\cite{LM21, LMSS22} may also conform to this template if the flooding parameters remain below these strict margins; 
a full analysis under randomized-oracle flooding (where the reaction bit becomes probabilistic) is left for future work.

\section{Experimental Validation}\label{sec:experiment}
We experimentally evaluate the attacks presented in \Cref{alg:on-attack} and~\ref{alg:algo_sparse} using 
Lattigo~v6.2.0~\cite{lattigo} across various parameter settings.

The implemented attack 
follows the steps of \Cref{alg:on-attack} and~\ref{alg:algo_sparse}.
The oracle decodes the resulting plaintext and returns \texttt{valid} if and only if $\|\mathsf{slots}\|_\infty\le 1$, matching the slot-domain reaction oracle of \Cref{sec: reaction-oracle}. 
The per-query plaintext multiplication does not consume the relinearization key, so the attack needs no $\mathsf{rlk}$ access. 

All experiments were executed using a single-thread
on an Intel Core i9-10940X {at $3.30$\,GHz, with $128$\,GB of RAM, under Ubuntu~20.04.6 LTS}.

\medskip\noindent\textbf{Parameters.}
We sweep five parameter triples $(\log N, \log q_0$, $\log \Delta)$ based on commonly used CKKS parameter regimes: $({12},35,30)$, $({12},37,32)$, $({14},45,40)$, $({16},55,50)$, and $({16},60,58)$, spanning low- to high-precision settings. 
They share ring dimensions and level-zero moduli with \Cref{tab: number of queries}, with varying $\log\Delta$ to examine the margin. 

For the secret key distribution, we primarily focus on the two standard configurations implemented in the Lattigo library: 
(i)~the uniform ternary distribution \texttt{Ternary(P:\,2/3)}, where each coefficient is non-zero with probability $2/3$ and takes a uniform sign, 
and (ii)~the fixed-weight sparse ternary distribution \texttt{Ternary(H:\,$h$)}, where the number of non-zero coefficients is fixed to a constant $h \ll N$. 
To demonstrate the comprehensive reach of our attack, we further extend it to additional distributions, such as Bernoulli-sparse ternary distributions (\texttt{Ternary(P:\,}$\cdot$\texttt{\,)}) and the discrete Gaussian distribution (\texttt{DiscreteGaussian(Sigma:\,$\cdot$\,, Bound:\,$\cdot\,$)}).

For these parameter regimes, we employ Lattigo’s six-prime modulus chains.
The {C2S} layer consumes four levels, leaving an output modulus of at least $q_0$, where we evaluate the modulus condition \Cref{eq: modulus}.
The mask multiplication applies plaintext multiplication without rescaling, so decoding by the ciphertext's own scale preserves the slot values.
The maximum modulus bit is then $\log Q \approx \log q_0 + 50\times5\approx285$--$305$, and $\log PQ \lesssim 427$ for most parameter sets, except for one parameter set at $\log N=16$ with $\log PQ\approx472$. The full configurations are provided in the artifact.

Although our attack is independent of the underlying RLWE hardness or the IND-CPA security of the scheme, the evaluated parameters for $\log N \in \{14, 16\}$ with the uniform ternary secrets conform to the standard 128-bit security level per the HE security guidelines~\cite{BCC+24}.

\begin{table*}[h!]
\centering
\small
\begin{tabular}{c|c|c|c|c|c|c|c}
\toprule
$\log N$ & $\log q_0$ & $\log\Delta$ & trials & $|\mathcal{S}|$ (mean$\pm$std, $|\mathcal{S}|/N$) & $\#Q$ (mean$\pm$std) & $\#Q/H_{2/3}(\vs)$ & wall (sec, mean$\pm$std) \\
\midrule
$12$ & $35$ & $30$ & $50$ & $2{,}736\pm 26$\,\,($0.668$) & $6{,}832\pm 26$ & $1.052$ & $30.02\pm 0.14$ \\
$12$ & $37$ & $32$ & $50$ & $2{,}730\pm 34$\,\,($0.667$) & $6{,}826\pm 34$ & $1.051$ & $30.06\pm 0.19$ \\\midrule
$14$ & $45$ & $40$ & $50$ & $10{,}929\pm 63$\,\,($0.667$) & $27{,}313\pm 63$ & $1.052$ & $506.1\pm 5.1$ \\\midrule
$16$ & $55$ & $50$ & $\phantom{0}5$ & $43{,}696\pm 141$\,\,($0.667$) & $109{,}232\pm 141$ & $1.052$ & $8{,}529\pm 168$ \\ %
$16$ & $60$ & $58$ & $\phantom{0}5$ & $43{,}683\pm 124$\,\,($0.667$) & $109{,}219\pm 124$ & $1.051$ & $68{,}408\pm 76^*$ \\
\bottomrule
\end{tabular}
\caption{End-to-end results of \Cref{alg:on-attack}, against Lattigo~v6.2.0 (with uniform ternary secrets). Secret keys were fully recovered in every trial. 
The first four parameters 
are standard CKKS regimes sharing $(N, q_0)$
with~\cite{lattigo};
the last one 
is in a high-precision setting~\cite{BCC+24}. 
The wall time in the last row (*), was obtained from five concurrent trials; hence, is not directly comparable to other single-trial rows.
$\#Q$ denotes the query-counts, whose theoretical lower bound is $H_{2/3}(\vs)$.}
\label{tab: experiment}
\end{table*}

\medskip\noindent\textbf{Results of \Cref{alg:on-attack}.}
We summarize the attack results targeting CKKS with uniform ternary secrets, in \Cref{tab: experiment}. 
Every trial recovered the secret in full: across all $160$ runs, in a total of $>3{,}000{,}000$ oracle queries, 
every coefficient was recovered exactly, with no false positives or false negatives from the reaction oracle. 
The two $\log N={16}$ regimes were each evaluated on five trials owing to their substantially higher runtime; yet together, they comprise over $1{,}000{,}000$ oracle queries with no error. 
By construction, the query count is $N+|\mathcal{S}|\approx 1.667\,N$.
This is within a factor $1.05$ of the average-case lower bound
$H_{2/3}(\vs)\approx 1.585\,N$ from \Cref{thm: lower-bound}.
End-to-end, the attack completed in $30.0$\,seconds at $\log N={12}$, $8.4$\,minutes at $\log N={14}$, and $142$\,minutes at $(16,55,50)$ on a single thread, with a per-trial standard deviation below $2\%$ of the mean. 
The dominant per-query cost was the plaintext multiplication and decryption.

\medskip\noindent\textit{Robustness across secret distributions.}
To verify that the attack is not tied to the default uniform ternary secret, we sweep additional secret distributions at $(12, 35, 30)$ and $(14, 45, 40)$: ternary with fixed-weights of $h\in\{128,192,256\}$, Bernoulli-sparse ternary with $p\in\{1/3,1/2,2/3,0.9\}$, and a discrete Gaussian with $\sigma=3.2$ truncated to the integer support $[-S,S]$, where $S=\lceil6\sigma\rceil=20$.
The results are summarized in \Cref{tab: distsweep}: 
All $515$ runs recovered the full secret. Each of the $490$ ternary-secret runs used $N+|\calS|<2N$ queries; the $25$ discrete-Gaussian runs used the bounded-secret bisection described in \Cref{rmk: bounded-secret}.

Against the average-case lower bound $H_p(\vs)$ of \Cref{thm: lower-bound}, the Bernoulli-sparse reported a factor of $\le1.07$ for $p\le 2/3$ and is essentially tight at $p=1/2$.
The discrete Gaussian secret is recovered with $\approx 5.4\,N$ queries, close to the support-size benchmark $N\log_2(2S+1)\approx5.36\,N$ for the bounded-secret bisection of Remark~\ref{rmk: bounded-secret}.
The fixed-weight cases show worse ratios against their exact entropy $H(\vs)=\log_2\binom{N}{h}+h$.

\begin{table}[t]
\centering
\small
\begin{tabular}{l|cc|cc}
\toprule
 & \multicolumn{2}{c|}{$\log N={12}$} & \multicolumn{2}{c}{$\log N={14}$} \\
Secret distribution & trials & $\#Q/H(\vs)$ & trials & $\#Q/H(\vs)$ \\
\midrule
ternary, $h=128$ & $50$ & $4.470$ & $20$ & $13.724~\,$ \\
ternary, $h=192$ & $50$ & $3.286$ & $20$ & $9.785$ \\
ternary, $h=256$ & $50$ & $2.666$ & $20$ & $7.728$ \\\midrule
ternary, $p=1/3$ & $50$ & $1.064$ & $20$ & $1.065$ \\
ternary, $p=1/2$ & $50$ & $1.000$ & $20$ & $0.999$ \\
ternary, $p=2/3$ & $50$ & $1.051$ & $20$ & $1.051$ \\
ternary, $p=0.9$ & $50$ & $1.388$ & $20$ & $1.388$ \\\midrule
Gaussian, $\sigma=3.2$ & $20$ & $~\,1.008^*$ & $\phantom{0}5$ & $~\,1.009^*$ \\
\bottomrule
\end{tabular}
\caption{Secret-distribution Robustness Tests at $(12, 35, 30)$ and $(14, 45, 40)$. 
The sweep covers fixed-weight ternary, Bernoulli-sparse ternary, and discrete-Gaussian secrets; every trial fully recovered the secret. 
$\#Q$ denotes the query count. For fixed-weight secrets, $H(\vs)=\log_2\binom{N}{h}+h$ is the exact entropy; for Bernoulli-sparse secrets, $H(\vs)=H_p(\vs)$. For the truncated discrete Gaussian (*), the denominator is the support-size benchmark $N\log_2(2S+1)$ with $S=20$, rather than its Shannon entropy.
}
\label{tab: distsweep}
\end{table}

\medskip\noindent\textit{Empirical noise.}
To empirically validate the noise behavior analyzed in our framework, we measured the exact noise profiles during the C2S layers, confirming that both the slot-margin and modulus conditions hold with sufficient margins.
In the coefficient domain, let $\vm^{(2)}$ denote the plaintext carried by the Step-2 ciphertext. 
The empirical error $\|\vm^{(2)} - \Delta\vs\|_\infty$ remains strictly below $2^{17.3}$ across the entire parameter sweep, safely falling under the conservative analytic bound of $B_{\mathsf{C2S}}^{\mathrm{coeff}}\le 2^{18}$ employed for the modulus condition in \Cref{eq: modulus}.

In the slot domain, the measured noise $B_{\mathsf{C2S}}^{\mathrm{slot}}$ tops out at
$\le 2^{-12.7}$ at $({12},35,30)$ and drops further to, for e.g., $\le 2^{-32.0}$ at $({16},55,50)$. 
Each reported value represents the worst-case $\ell_\infty$ maximum captured across all queried coordinates and experimental trials. 
Under a constant amplitude of $\alpha=4$, the slot-margin scaling term satisfies $\alpha B_{\mathsf{C2S}}^{\mathrm{slot}}\le 2^{-10.7}\ll\tau=1$, even under the lowest precision configuration, and decreases by several additional bits when $\log\Delta\ge 30$. 
This empirically validates 
\Cref{eq: alpha-margin}. 

\begin{table}[t]
\centering\small
\begin{tabular}{cc|cc|c}
\toprule
$\log N$ & $h$ & $\#Q$ & $\#Q/H(\vs)$ & LLL Time (sec.) \\
\midrule
$12$ & $32$  & $338$     & $1.134$ & $0.4$ \\
$12$ & $64$  & $598$     & $1.117$ & $3.3$ \\
$12$ & $128$ & $1{,}079$ & $1.142$ & $1{,}521~~$ \\\midrule
$14$ & $128$ & $1{,}469$ & $1.221$ & $\,1{,}521^*~$ \\\midrule
$16$ & $128$ & $1{,}801$ & $1.234$ & $\,1{,}521^*~$ \\
\bottomrule
\end{tabular}
\caption{End-to-End Results of \Cref{alg:algo_sparse}
(with ternary secrets of fixed-weight $h$).
Every run recovered the exact support; then LLL fully recovered the secret keys.
Single-core LLL time (in seconds, $q=\min q_0$) was measured only at $\log N={12}$, since it only depends on $h$ ($*$). 
$\#Q$ denotes the query count, and $H(\vs)=\log_2\binom{N}{h}+h$ is the exact entropy of a uniformly sampled fixed-weight ternary secret.
}
\label{tab: grouptest}
\end{table}

\medskip
\noindent\textbf{Results of \Cref{alg:algo_sparse}.}
\Cref{tab: grouptest} summarizes the attack results based on \Cref{alg:algo_sparse}.
Across $h\in\{32,64,128\}$ and $\log N \in \{12, 14, 16\}$,
it recovers the exact support, with no false positives or negatives, in $O(h\log(N/h))$ reaction queries: e.g., at $(\log N,h)=({16},128)$, it requires only $1{,}801$ queries, achieving $36\times$ reduction over 
the $N+h=65{,}664$ queries necessitated by 
\Cref{alg:on-attack}. 

The recovered support restricts the LWE instance to a genuine $h$-dimensional LWE problem that can be solved directly from the public key material. Specifically, applying the Lenstra-Lenstra-Lov\'{a}sz (LLL) algorithm~\cite{AC:DucPulSte25,lenstra1982factoring} on a lattice of dimension $2h+1$ with $N$ samples efficiently and exactly recovered every secret sign for $h \in \{32, 64, 128\}$. 
The single-core solving time scales steeply with the Hamming weight, taking $0.4$\,seconds for $h=32$,
$3.3$\,seconds for $h=64$, and 
$25$\,minutes for $h=128$, when $\log q=\min \log q_0 \approx 35$. 

According to the lattice estimator~\cite{APS15}, under the default ternary-secret model, the estimated bit-cost for solving the reduced lattice instances grows from $33.1$ bits at $h=32$ to $43.2$ bits at $h=1{,}024$, for $\log q = 285$ (the minimum
\pk\ modulus $Q$). 
Since we fix the modulus to small $\approx285$ bits,
these are rough upper bounds: for larger moduli of $700$--$1{,}600$ bits and for standard Hamming weights of $h\in\{128,192,256\}$~\cite{BTPH22,BCC+24}, the cost further decreases, far below the standard $128$-bit level.

\section{Countermeasures}\label{sec: discussion}

{%
Reaction leakage can be addressed at several layers. Below, we present two direct and deployment-compatible measures: standardizing the externally visible response and enforcing per-key reaction monitoring. 

\medskip\noindent\textbf{Removing the reaction signal.}
The oracle exists only because the client's externally visible behavior depends on the validity outcome, so the most robust and deployment-compatible mitigation is to remove that dependence: return the same protocol response regardless of the admissibility check, handling any failure internally. This eliminates the one-bit oracle while changing only outward behavior, not the FHE computation. 

\medskip\noindent\textbf{Bounding and monitoring reactions.}
When the response cannot be fully standardized, bounding the number of reactions per secret key is a cheap complementary defense. Key recovery requires resolving enough coordinates that the projected LWE instance falls into lattice-reduction range (e.g., LLL); a sufficiently small cap will leave enough room for security. Even without a hard cap, the attack is detectable: each adversarial trial forces a rejection with high probability, so an elevated per-key rejection rate and query volume provide a natural aggregate detection vector.

}

\medskip
Other than the two major directions, we consider three further mitigation strategies. Below, we discuss why each method---despite seeming straightforward to apply---presents its own limitations, rendering them either ineffective as a standalone defense or detrimental to correctness, performance, or efficiency.

\medskip\noindent\textbf{Noise and Parameter Modifications.}
Violating the slot-margin condition (\Cref{eq: alpha-margin}) via parameter tuning is fragile: because this margin is tied to CKKS correctness, increased noise inherently risks rejecting honest evaluations. Protective randomization must therefore target the observable reaction directly, rather than the homomorphic computation pipeline. Crucially, standard IND-CPA$^\text{D}$ smudging~\cite{GNS+24,LM21, LMSS22} does not transfer to our setting. Smudging typically adds noise $\|\ve_{\mathsf{smdg}}\|_\infty \ll \Delta$ to mask sub-$\Delta$ structures. However, against \Cref{alg:on-attack}, our isolated slot for $s_i \neq 0$ reaches magnitude $\alpha=4$, quadruple the admissible range $\tau=1$. Pulling this manipulated slot below $\tau$ requires smudging noise of magnitude at least $\alpha-\tau=3$, which would indiscriminately reject honest outputs. In short, smudging strong enough to thwart our attack is functionally indistinguishable from smudging that destroys correctness, so parameter changes are impractical for deep or iterative workloads.

\medskip\noindent\textbf{Changing Rings.}
Our attack relies strictly on the encryption-key shift $(\va-\Delta,\vb)$ and a public C2S transform, remaining entirely agnostic to the splitting behavior of $x^N+1$. Migrating from power-of-two cyclotomics to alternative rings~\cite{CLS17,LS15} is therefore not a standalone defense if an analogous transform exists. Furthermore, this approach severely degrades performance by sacrificing highly optimized NTT pipelines.

\medskip\noindent\textbf{Cryptographic Integrity.}
Upgrading to vCCA-secure FHE~\cite{LMSV12, MN24} or verifiable FHE~\cite{ABPS24, BCFK21, FGP14, FNP20, GNS23, GGP10, GKP+13} provides a principled defense, removing the reaction signal by rejecting ciphertexts that are not certified outputs of an intended computation. Recent CKKS-specific verifiable schemes~\cite{CCC+25, ZWL+25} would cleanly reject our crafted ciphertexts, as the attacker cannot prove honest execution of a client-approved circuit. However, these mechanisms introduce a substantial efficiency penalty, making them difficult to integrate into practical deployments.

\medskip
Overall, the standardization of reaction and per-key reaction monitoring remains the most practical mitigation.

\section{Conclusion and Open Questions}\label{sec: conclusion}
We presented a protocol-level reaction attack against CKKS deployments that expose decoded-slot validity responses and C2S rotation keys. Using a simple public-key shift and fresh encryptions of zero, our attack recovers ternary secrets in query complexity that tightly matches theoretical lower bounds.
Crucially, by independently randomizing each query, our approach inherently bypasses simple pattern-matching defenses.

Our results demonstrate that vulnerabilities can emerge strictly at the protocol level, without refuting RLWE or IND-CPA security: secure outsourced computation must tightly control how decoded outputs, including single-bit validity reactions, are exposed. The most practical defense is to standardize the client's externally visible response and monitor reactions per key. The elevated query volume and rejection rate our attack induces further offer a natural detection vector through aggregate failure-rate analysis.

\medskip\noindent\textbf{Open Questions.}
We highlight four directions for future work:
(i) Detection Models: Formalizing detector models to determine which aggregate heuristics (e.g., query volume, rejection rates) provably prevent key recovery; 
(ii) Closing Theoretical Gaps: Reducing the minimal query gap between our theoretical bound and the experimental average for sparse secrets (close to the $\approx 1.05$ gap in uniform ternary); 
(iii) Scheme Generality: Investigating whether similar reaction attacks threaten other widely deployed RLWE schemes, such as BGV and BFV; and 
(iv) Safe Validation: Designing output-validation mechanisms that support legitimate retries without leaking externally exploitable accept/reject bits.

\section*{Acknowledgments}
Jiseung Kim is supported by the National Research Foundation of Korea (NRF) grant funded by the Korean government (MSIT) (No. RS-2026-25477662), and the Institute of Information \& Communications Technology Planning \& Evaluation~(IITP) grant funded by the Korea government(MSIT) (No. RS-2024-00399491, Development of Privacy-Preserving Multiparty Computation Techniques for Secure Multiparty Data Integration).
Changmin Lee is supported by the Korea Internet \& Security Agency (KISA) grant funded by the Korean government (PIPC) (No. 2780000036, Masters and Doctoral Program for Advanced Talent Development in Personal Data Protection and Responsible Use) and by Korea University grant K2603241. Song Bian is supported by 
the National Natural Science Foundation of China (62572020).

\bibliographystyle{IEEEtranS}
\bibliography{paper}

\appendices
\section{Ethics Considerations}
This paper presents a key-recovery attack against a class of CKKS deployments that expose a decoded-output validity reaction together with public evaluation material sufficient for a coefficient-to-slot layer. Because the work develops a new attack technique, it carries misuse risk: an adversary could use the technique to compromise confidentiality in CKKS-based outsourcing protocols that expose such reaction signals.

The risk to real users in our study is limited. All experiments were conducted in controlled local environments using locally generated keys, synthetic data, and local Lattigo instances. We did not target, measure, exploit, or interact with third-party or production systems. We did not collect personal data, use human-subject data, or involve human participants.

The vulnerability we identify is a deployment-surface and protocol-level leakage channel, not a memory-safety or implementation bug in a specific library. The relevant stakeholders include maintainers of FHE libraries such as Lattigo, OpenFHE, and HEaaN, as well as organizations deploying CKKS-based outsourcing protocols whose application surface exposes a per-query validity reaction. We followed coordinated vulnerability disclosure (CVD): we notified the relevant library maintainers and shared the artifact privately before the public release of exploit-ready code. 

The paper also provides mitigation guidance. Primary operational mitigations include standardizing externally visible protocol responses to remove the reaction signal entirely and enforcing per-key reaction caps or monitoring anomalous rejection rates. We also analyze three alternative strategies---parameter modifications (including smudging noise), changing cyclotomic rings, and cryptographic integrity mechanisms (vCCA or verifiable FHE)---and explain their limitations in terms of correctness, performance, or computational overhead.

We believe the defensive benefits of publication outweigh the misuse risk. The attack was validated without affecting real users or systems, and the paper clarifies a concrete deployment risk together with actionable mitigation for CKKS deployments.

\end{document}